\documentclass[12pt]{article}
\usepackage{algorithm}
\usepackage[noend]{algpseudocode}
\usepackage{mathtools}
\usepackage{amsthm,amsmath,amssymb}
\usepackage{subcaption}
\usepackage{times}
\usepackage{newtxtext}
\usepackage{graphicx}
\usepackage{multirow,booktabs,array}
\usepackage{tabularx,enumerate}
\usepackage[left]{lineno} 
\usepackage{float}

\usepackage{tikz}

\newcommand{\bigwhitestar}{%
  \begin{tikzpicture}[scale=0.09, baseline=-0.5ex]
    \draw (90:1) -- (234:1) -- (18:1) -- (162:1) -- (306:1) -- cycle;
    \draw (90:1) -- (162:1) -- (234:1) -- (306:1) -- (18:1) -- cycle;
  \end{tikzpicture}%
}

\newtheorem{theorem}{Theorem}[section]
\newtheorem{lemma}{Lemma}[section]
\newtheorem{corollary}{Corollary}[section]
\newtheorem{claim}{Claim}

\newtheorem{definition}{Definition}

\newcommand{\yes}{\mbox{Yes}}
\newcommand{\no}{\mbox{No}}
\newcommand{\yesins}{{\yes}-instance}
\newcommand{\noins}{{\no}-instance}

\newcommand{\np}{\textsf{\mbox{NP}}}
\newcommand{\nph}{{\np}-{hard}}

\newcommand{\npc}{{\np}-{complete}}
\newcommand{\npcns}{\np-completeness}

\newcommand{\prob}[1]{{\sc{#1}}}
\newcommand{\vset}[1]{V(#1)}
\newcommand{\eset}[1]{E(#1)}
\newcommand{\neighbor}[2]{N_{#2}(#1)}
\newcommand{\degree}[2]{{\textsf{deg}}_{#2}({#1})}
\newcommand{\edge}[2]{\{#1,#2\}}
\newcommand{\abs}[1]{|#1|}
\newcommand{\setmid}{:}
\newcommand{\term}[1]{{\textit{#1}}}
\newcommand{\prealg}[2]{
	\begin{tabular}{l p{0.85\textwidth}}
		{\bf{Input:}} & #1.\\
		{\bf{Output:}}&  #2.\\ 
	\end{tabular}
}
\newcommand{\tp}[2]{\textsf{top}(#1,#2)}  
\newcommand{\ttp}[1]{\textsf{top}(#1)}
\newcommand{\second}[2]{\textsf{second}(#1,#2)}  
\newcommand{\bottom}[2]{\textsf{last}(#1,#2)}  
\newcommand{\attg}[1]{\textsf{AttDiG}(#1)}
\newcommand{\attug}[1]{\textsf{AttG}(#1)}
\newcommand{\parent}[2]{\textsf{parent}(#1,#2)}  
\newcommand{\forced}{forced}
\newcommand{\free}{free}
\newcommand{\sbo}[2]{#1_{[\dots #2]}} 

\newtheorem{observation}{Observation}
\newtheorem*{purificationrule}{Purification Rule}
\newtheorem{openquestion}{Open Question}

\newcommand{\guide}[2]{\textsf{guider}(#2,#1)}

\newcommand{\dfssu}[1]{T^{\text{DFS}}_{#1}}

\newcommand{\EPP}[3]
{\begin{center}
		{\small
			\begin{tabularx}{0.85\columnwidth}{ll}
				\toprule
				\multicolumn{2}{l}{\textsc{#1}} \\ \midrule
				{\bf Input:}   & \parbox[t]{0.8\columnwidth}{#2\vspace*{1mm}}  \\
				{\bf Question:}& \parbox[t]{0.8\columnwidth}{#3\vspace*{.5mm}} \\ \bottomrule
			\end{tabularx}
		}
\end{center}}

\DeclareMathOperator*{\argmin}{arg\,min}

\title{When Graph Traversal Meets Structured Preferences: Unified Framework and Complexity Results}
\author{Guozhen Rong$^1$, Xin Li$^1$, Yongjie Yang$^2$}
\date{$^1$School of Computer and Communication Engineering, Changsha University of Science and Technology, Changsha 410083, China\\
$^2$Department of Economics, Saarland University, Saarbr\"{u}cken 66123, Germany}

 \usepackage{geometry}
\begin{document}
\maketitle

\begin{abstract}
Preference restrictions have played a significant role in computational social choice. This paper studies a framework that connects preference restrictions with classical graph search paradigms. We model candidates as vertices of a graph and interpret each voter's preference ordering as the outcome of traversing the graph according to a graph search. We focus on six fundamental paradigms: breadth-first search (BFS), depth-first search (DFS), breadth-first search (LexBFS), lexicographic depth-first (LexDFS), maximum cardinality search (MCS), and maximal neighborhood search (MNS).

Within this framework, we study the problem of determining whether a given preference profile admits a graph support subject to structural restrictions, that is, whether there exists a graph such that each preference ordering can be generated by traversing the graph under the chosen paradigm. For all considered paradigms, we show that this problem is NP-hard when the graph support is required to have at most~$k$ edges. We further extend these hardness results to the case where the graph support is required to have maximum degree~$k$.
For DFS, we prove that recognizing whether a preference profile admits a tree support can be solved in polynomial time. Existing results imply polynomial-time solvability of the problem for all remaining graph traversals, except BFS and LexBFS, for which the complexity remains open.
\end{abstract}

\section{Introduction}
\subsection{Background}
The study of structured preferences is a central theme in computational social choice.
Classical domain restrictions such as single-peakedness~\cite{Black1948}, single-crossingness~\cite{Mirrlees71,Roberts1977329}, and group-separability~\cite{Inada1964,Inada1969} have been extensively studied, as they not only render many otherwise intractable computational problems polynomial-time solvable and provide a way to circumvent classical impossibility theorems, but also capture structural properties observed in many real-world applications~\cite{DBLP:journals/jair/BrandtBHH15,DBLP:conf/ijcai/CornazGS13,DBLP:journals/iandc/FaliszewskiHHR11,DBLP:conf/atal/Yang15,DBLP:journals/jcss/YangG17}.
Generalizations such as single-peakedness on trees~\cite{Demangesinglepeaked82,DBLP:conf/aaai/PetersE16,MichaeTricksinglepeaked89} and on circles~\cite{DBLP:journals/jair/PetersL20}, and single-crossingness on median graphs~\cite{DBLP:conf/ijcai/ClearwaterPS15} have further broadened the scope of this approach. A common feature of these domains is that they impose an underlying structure---often a graph from a specific class---that constrains how preferences can be organized.
The substantial body of existing work naturally invites a broader reflection:
\begin{quote}
    {\textit{Perhaps the time has come to examine graph-based preferences in a more general and systematic way?}}
\end{quote}

In this context, a graph search paradigm (a.k.a.~graph traversal), which specifies the order in which the vertices of a graph are visited, provides a natural framework for addressing this question.
Graph search paradigms such as generic search (GS), breadth-first search (BFS), depth-first search (DFS), lexicographic breadth-first search (LexBFS), lexicographic depth-first search (LexDFS), maximum cardinality search (MCS), and maximal neighborhood search (MNS) are among the most fundamental tools in graph theory and algorithm design~\cite{DBLP:journals/siamdm/BretscherCHP08,DBLP:journals/siamcomp/RoseTL76,DBLP:journals/siamcomp/Tarjan72}.
The generic search GS captures the weakest requirement, ensuring that at every step the set of visited vertices induces a connected subgraph. Viewed from a preference-theoretic perspective, this requirement naturally corresponds to the idea that preferences evolve through connected local comparisons.
Other paradigms can be seen as refinements of GS that introduce additional ordering principles.
Consequently, for a fixed graph, any ordering obtained by traversing the graph according to one of these paradigms is also a GS-ordering of the same graph.

While graph traversals were originally developed to characterize important graph classes and to support efficient recognition algorithms, their use has expanded far beyond these classical purposes, and their potential for further applications continues to grow.
For instance, BFS and DFS play a central role in querying and reasoning over knowledge graphs, where they are employed to explore semantic neighborhoods and compute path-based similarities~\cite{DBLP:conf/icde/YangHWY16,DBLP:conf/jist/ZhuRLWTY15}.
Graph search paradigms also underpin many areas of artificial intelligence, including heuristic search algorithms~\cite{KORF198597} and reasoning in abstract argumentation systems~\cite{NOFAL201423}.
Fundamentally, these diverse applications benefit from the fact that graph traversals can be viewed procedurally as models of sequential exploration.
This renders them a natural model for generating structured preference orderings from an underlying graph of alternatives. The connection is further supported by recent developments in decision-making frameworks, which increasingly model alternatives (candidates) as vertices of a graph, and relations such as interdependence, similarity, or interaction as edges~\cite{AMGOUD2014585,KleinerMo2020,DBLP:journals/scw/ReyEH25,DBLP:conf/ijcai/YangW18}.

\subsection{Our Contributions}
In exploring the above question, the work of Escoffier, Spanjaard, and Tydrichova~\cite{DBLP:journals/dam/EscoffierST24} constitutes an early step in this direction.
Given a set of orderings, they studied whether there exists a graph with at most~$k$ edges (or maximum degree at most~$k$) such that every prefix of every ordering induces a connected subgraph, and proved this problem to be {\nph}.
Although not stated explicitly, their notion naturally corresponds to the generic search paradigm.

We generalize this setting by considering the problems \prob{Edge-Bounded-$\mathcal{S}$-SP} and \textsc{Deg-Bounded-$\mathcal{S}$-SP}, which ask whether there exists a graph with at most~$k$ edges (or maximum degree at most~$k$) that admits all given orderings as $\mathcal{S}$-orderings, for a graph search paradigm~$\mathcal{S}$.
We complement the result of~\cite{DBLP:journals/dam/EscoffierST24} by showing that, for $\mathcal{S}\in\{\textnormal{DFS},\textnormal{BFS},\textnormal{LexDFS},\textnormal{LexBFS},\textnormal{MNS},\textnormal{MCS}\}$, both problems are {\nph}.

We also study the variant \textsc{SP-$\mathcal{S}$-Tree}, where the support is required to be a tree.
For GS, this coincides with the classical problem of recognizing single-peakedness on a tree, which is polynomial-time solvable~\cite{Demangesinglepeaked82,DBLP:journals/dam/EscoffierST24,MichaeTricksinglepeaked89,PetersYCE22}.
Since GS, MNS, and MCS coincide on trees, \textsc{SP-$\mathcal{S}$-Tree} is also polynomial-time solvable for $\mathcal{S}\in\{\textnormal{MNS},\textnormal{MCS}\}$.
Moreover, LexBFS coincides with BFS on trees, and LexDFS with DFS, leaving BFS and DFS as the only remaining open cases.
We resolve the DFS case by presenting a polynomial-time algorithm and, in doing so, establish several structural properties of DFS on trees that may be of independent interest.
The complexity of the BFS case remains open.

\subsection{Motivation}
While the preceding sections establish the theoretical context and our formal contributions,
the following discussion highlights the broader motivations of our framework.

In many real-world applications, preferences are far from arbitrary.
Instead, they often exhibit strong structural regularities that can be captured by well-defined domains such as single-peaked or single-crossing preferences.
These structural restrictions are of practical and theoretical importance:
on the one hand, they arise naturally in applications ranging from political elections (where candidates can often be placed along an ideological spectrum) to resource allocation and matching theory~\cite{DBLP:journals/aamas/BeynierMRS21,Bonifacio2015Socialchoicewelfsinglepeaked,RAD_ROY_2021};
on the other hand, they enable efficient algorithms in settings where general preference aggregation is computationally intractable~\cite{DBLP:journals/jair/BrandtBHH15,DBLP:conf/ijcai/CornazGS13,DBLP:journals/iandc/FaliszewskiHHR11,DBLP:conf/atal/Yang15,DBLP:journals/jcss/YangG17}.

Our work is motivated by the observation that preferences can be interpreted as the outcome of a \emph{sequential exploration} of alternatives embedded in a graph in certain scenarios.
For example, in decision-making scenarios, an agent may start from a known alternative and explore others step by step, comparing each new alternative to those already considered.
This process closely mirrors classical graph search paradigms and motivates the study of a unified framework in which preferences arise from traversals of an underlying graph of alternatives.

This connection is both conceptually and practically appealing.
Graph search provides a rich toolbox of paradigms that model different exploration strategies.
For instance, DFS corresponds to a ``deep dive'' into one branch of the search space before backtracking, while BFS represents a more cautious exploration level by level.
More sophisticated searches such as MCS or MNS are designed may be interpreted as plausible models of decision-making in contexts where agents prioritize alternatives with larger or more informative connections.
In MCS, vertices are visited in an order that successively maximizes the number of previously explored neighbors, whereas MNS prioritizes vertices whose set of visited neighbors is inclusion-wise maximal.
Both procedures capture a notion of exploration guided by connectivity or information richness, rather than simple connectivity.

In addition to the aforementioned potential applications, our study is further motivated by the generality of our model, which encompasses several classical domains as special cases.
For instance, a set of orderings is single-peaked if and only if it admits a GS-support that is a path.
Since any set of orderings on a set~$V$ admits the clique on~$V$ as a support, the model naturally bridges many restricted preference domains with the general domain.
Further connections between our framework and previously studied domains will be discussed in the next section.

\section{Preliminaries}
For an integer $i$, let $[i] = \{1, 2, \ldots, i\}$ denote the set of positive integers no greater than $i$. A {\term{pair}} is a set of cardinality two. For a set $X$, let $\overrightarrow{X}$ denote any arbitrary fixed ordering of $X$, unless stated otherwise.

\subsection{Graphs}
For a graph $G=(V, E)$, we also denote its vertex set and edge set by~$\vset{G}$ and~$\eset{G}$, respectively. An edge between two vertices $u$ and $v$ is denoted as $\edge{u}{v}$. We write~$\neighbor{v}{G}$ for the set of {\term{neighbors}} of~$v$ in~$G$, that is, $\neighbor{v}{G}=\{u \setmid \edge{u}{v}\in E\}$,
and let~$\degree{v}{G}=\abs{\neighbor{v}{G}}$ denote the degree of~$v$ in~$G$. When the graph under consideration is clear from the context, we omit~$G$ from the notation. 
Vertices of degree one (respectively, zero) are referred to as {\term{pendant}} (respectively, {\term{isolated}}) vertices. If a pendant vertex~$v$ is adjacent to~$u$, then~$v$ is called a {\term{pendant}} of~$u$. For a vertex subset~$S\subseteq\vset{G}$, we denote by~$G-S$ the graph obtained from~$G$ by deleting all vertices in~$S$ (together with their incident edges). When $S=\{v\}$ is a singleton, we simply write~$G-v$.

\subsection{Graph Search Paradigms}

Let $\sigma$ be an ordering of a set $V$. By convention, we write the elements of~$\sigma$ from left to right. For $u, v \in V$, if~$u$ appears to the left of~$v$ in~$\sigma$, we say that~$u$ \term{precedes}~$v$, or equivalently, that~$v$ \term{succeeds}~$u$. We denote this relation by $u <_\sigma v$. For~$u, v, w\in V$ with $u <_\sigma v$ and $v <_\sigma w$, we may write $u <_\sigma v <_\sigma w$, and say that~$v$ lies between~$u$ and~$w$ in~$\sigma$. This notation naturally extends to sequences of arbitrary length.
The \emph{position} of an element~$v \in V$ in~$\sigma$, denoted~$\sigma(v)$, is one more than the number of elements that precede~$v$ in~$\sigma$. We may represent an ordering by listing its elements in parentheses; for example, $(a,b,c)$ denotes an ordering in which~$a$,~$b$, and~$c$ have positions~$1$,~$2$, and~$3$, respectively.

Let $n=\abs{V}$. For an integer~$i \in [n]$, we use $\sigma_{\leq i}$ to denote~$\sigma$ induced by its first~$i$ elements. We call~$\sigma_{\leq i}$ a \emph{prefix} of~$\sigma$, and refer to the remaining part as its \emph{suffix}. For an element~$u \in V$, $\sbo{\sigma}{u}$ denotes the prefix of~$\sigma$ ending at~$u$.
For a subset~$S \subseteq V$, we write~$\sigma \ominus S$ for the ordering obtained from~$\sigma$ by deleting all vertices in~$S$.
For example, if $\sigma=(c,a,b,v,d)$, then $\sigma_{\le 3}=\sbo{\sigma}{b}(c,a,b)$, and $\sigma\ominus\{a,v\}=(c,b,d)$.

A \emph{graph search paradigm} is an algorithm that visits the vertices of a graph one by one, starting from an arbitrary initial vertex and following a specified rule to choose subsequent vertices. We consider six widely studied graph search paradigms. For each, we give an intuitive explanation and present formal pseudocode in Figure~\ref{fig-pesudocode}. We also characterize these paradigms using four-point orderings established by Corneil and Krueger~\cite{DBLP:journals/siamdm/CorneilK08}.

\begin{figure}
\noindent\begin{minipage}{0.5\textwidth}
BFS

\fbox{%
\begin{minipage}{0.95\linewidth}
\begin{algorithmic}[1]
\ForAll{$v \in \vset{G}$} \State $\mathrm{label}(v) \leftarrow 0$;
\EndFor \For {$i \leftarrow 1$ \textbf{to} $n$}
\State $v \leftarrow$ an unvisited vertex with the largest label;
\State $\sigma(v) \leftarrow i$;
\ForAll{unvisited $u \in \neighbor{v}{G}$}
\State $\mathrm{label}(u) \leftarrow \mathrm{max}(\mathrm{label}(u), n - i)$;
\EndFor
\EndFor
\State
\Return $\sigma$;
\end{algorithmic}
\end{minipage}
}
\end{minipage}\begin{minipage}{0.5\textwidth}
DFS

\fbox{%
\begin{minipage}{0.95\linewidth}
	\begin{algorithmic}[1]
		\ForAll{$v \in \vset{G}$}
		\State $\mathrm{label}(v) \leftarrow 0$;
		\EndFor
		\For {$i \leftarrow 1$ \textbf{to} $n$}
		\State $v \leftarrow$ an unvisited vertex with the largest label;
		\State $\sigma(v) \leftarrow i$;  
		\ForAll{unvisited $u \in \neighbor{v}{G}$}
		\State $\mathrm{label}(u) \leftarrow i$;
		\EndFor
		\EndFor
		\State \Return $\sigma$;
	\end{algorithmic}
\end{minipage}
}
\end{minipage}
\bigskip

\noindent\begin{minipage}{0.5\textwidth}
LexBFS

\fbox{%
\begin{minipage}{0.95\linewidth}
\begin{algorithmic}[1]
		\ForAll{$v \in \vset{G}$}
		\State $\mathrm{label}(v) \leftarrow 0$;
		\EndFor
		\For {$i \leftarrow 1$ \textbf{to} $n$}
		\State $v \leftarrow$ an unvisited vertex with the lexicographically largest label;
		\State $\sigma(v) \leftarrow i$; $j \leftarrow n-i$;  
		\ForAll{unvisited $u \in \neighbor{v}{G}$}
		\State append $j$ to $\mathrm{label}(u)$;
		\EndFor
		\EndFor
		\State \Return $\sigma$;
	\end{algorithmic}
\end{minipage}
}
\end{minipage}\begin{minipage}{0.5\textwidth}
LexDFS

\fbox{%
\begin{minipage}{0.95\linewidth}
\begin{algorithmic}[1]
		\ForAll{$v \in \vset{G}$}
		\State $\mathrm{label}(v) \leftarrow 0$;
		\EndFor
		\For {$i \leftarrow 1$ \textbf{to} $n$}
		\State $v \leftarrow$ an unvisited vertex with the lexicographically largest label;
		\State $\sigma(v) \leftarrow i$;  
		\ForAll{unvisited $u \in \neighbor{v}{G}$}
\State prepend $i$ to $\mathrm{label}(u)$;
		\EndFor
		\EndFor
		\State \Return $\sigma$;
	\end{algorithmic}
\end{minipage}
}
\end{minipage}
\bigskip

\noindent\begin{minipage}{0.5\textwidth}
MCS

\fbox{%
\begin{minipage}{0.95\linewidth}
\begin{algorithmic}[1]
		\ForAll{$v \in \vset{G}$}
		\State $\mathrm{label}(v) \leftarrow 0$;
		\EndFor
		\For {$i \leftarrow 1$ \textbf{to} $n$}
		\State $v \leftarrow$ an unvisited vertex with the largest label;
		\State $\sigma(v) \leftarrow i$;  
		\ForAll{unvisited $u \in N(v)$}
		\State $\mathrm{label}(u) \leftarrow \mathrm{label}(u) + 1$;
		\EndFor
		\EndFor
		\State \Return $\sigma$;
	\end{algorithmic}
\end{minipage}
}
\end{minipage}\begin{minipage}{0.5\textwidth}
MNS

\fbox{%
\begin{minipage}{0.95\linewidth}
\begin{algorithmic}[1]
		\ForAll{$v \in \vset{G}$}
		\State $\mathrm{label}(v) \leftarrow \emptyset$;
		\EndFor
		\For {$i \leftarrow 1$ \textbf{to} $n$}
		\State $v \leftarrow$ an unvisited vertex whose label is inclusion-wise maximal;
		\State $\sigma(v) \leftarrow i$; 
		\ForAll{unvisited $u \in N(v)$}
		\State $\mathrm{label}(u) \leftarrow \mathrm{label}(u) \cup \{i\}$;
		\EndFor
		\EndFor
		\State \Return $\sigma$;
	\end{algorithmic}
\end{minipage}
}
\end{minipage}
\caption{Pseudocode for several graph search paradigms.
Each algorithm takes as input a graph~$G$ and outputs an ordering~$\sigma$ of~$V(G)$.}
\label{fig-pesudocode}
\end{figure}

\begin{description}
    \item[Breadth-First Search (BFS).]
BFS explores the vertices of a graph in order of their distance from a selected vertex. It operates by visiting all vertices at distance $d$ from the start vertex before visiting any vertex at distance $d+1$.

\item[Depth-First Search (DFS).]
Starting from an initial vertex, DFS repeatedly visits an unvisited neighbor and continues recursively. When no such neighbor exists, it backtracks to the most recent vertex with unvisited neighbors.

\item[Lexicographic Breadth-First Search (LexBFS).]
LexBFS is a refinement of BFS that uses additional information to break ties between vertices that are eligible to be visited next%
. As in BFS, vertices are explored layer by layer, but when several unvisited vertices are available, LexBFS prefers those that are adjacent to vertices visited earlier in the search. Intuitively, vertices that have stronger connections to earlier explored parts of the graph are visited first, resulting in a more structured traversal than standard BFS.

\item[Lexicographic Depth-First Search (LexDFS).]
LexDFS is a refinement of DFS that uses additional information to break ties when several unvisited vertices are eligible to be explored next. As in DFS, the search proceeds by extending the current path as far as possible before backtracking. When multiple unvisited vertices are available, LexDFS prioritizes those that are adjacent to vertices visited more recently in the search. Intuitively, this reinforces the depth-first behavior by favoring continuations that stay close to the most recently explored parts of the graph.

 \item[Maximum Cardinality Search (MCS).]
    At each iteration, MCS selects an unvisited vertex maximizing the number of already visited neighbors.

    \item[Maximal Neighborhood Search (MNS).]
     At each iteration, MNS selects an unvisited vertex whose set of visited neighbors is inclusion-wise maximal among all unvisited vertices. 
\end{description}

For a graph~$G$ and a graph search paradigm~$\mathcal{S}$, an ordering~$\sigma$ of~$V$ is called an \emph{$\mathcal{S}$-ordering} if it is produced by a traversal of~$G$ according to~$\mathcal{S}$. A prefix of an $\mathcal{S}$-ordering is called a \emph{prefix $\mathcal{S}$-ordering} of~$G$.

\subsection{The Four-Point Ordering Characterizations}
Corneil and Krueger~\cite{DBLP:journals/siamdm/CorneilK08} established \emph{four-point ordering characterizations} for several graph search paradigms, expressing each traversal via simple local conditions on four vertices.

Let~$G=(V, E)$ be a graph and let $\sigma$ be an ordering of~$V$.
We say that~$\sigma$ \emph{satisfies} a property if, for every
$a,b,c\in V$ with $a <_{\sigma} b <_{\sigma} c$,
the corresponding condition holds.

\begin{description}
 \item[Property~S:] If $\edge{a}{c}\in E$ and $\edge{a}{b}\not\in E$, then there exists a vertex $d\in V$ with $d<_{\sigma}b$ and  $\edge{d}{b}\in E$.

    \item[Property~B:] If $\edge{a}{c}\in E$ and $\edge{a}{b}\not\in E$, then there exists a vertex $d\in V$ with $d<_{\sigma}a$ and  $\edge{d}{b}\in E$.

    \item[Property~D:] If  $\edge{a}{c}\in E$ and $\edge{a}{b}\not\in E$, then there exists a vertex $d\in V$ with $a<_{\sigma}d<_{\sigma}b$ and $\edge{d}{b}\in E$.

    \item[Property LB:] If  $\edge{a}{c}\in E$ and $\edge{a}{b}\not\in E$, then there exists a vertex $d\in V$ with $d<_{\sigma}a$ such that  $\edge{d}{b}\in E$ and $\edge{d}{c}\not\in E$.

    \item[Property LD:]  If $\edge{a}{c}\in E$ and $\edge{a}{b}\not\in E$, then there exists a vertex $d\in V$ with $a<_{\sigma}d<_{\sigma}b$ such that $\edge{d}{b}\in E$ and $\edge{d}{c}\not\in E$.

    \item[Property M:] If $\edge{a}{c}\in E$ and $\edge{a}{b}\not\in E$, then there exists a vertex $d<_{\sigma}b$ such that $\edge{d}{b}\in E$ and $\edge{d}{c}\not\in E$.
\end{description}

Figure~\ref{fig-properties} illustrates the above defined properties. 

\begin{figure}
    \centering
    \includegraphics[width=0.95\linewidth]{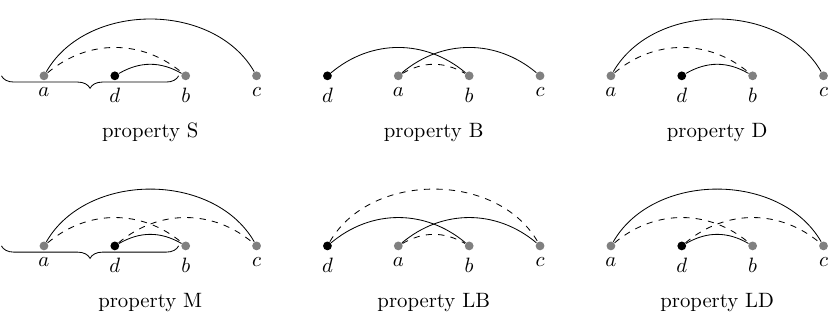}
    \caption{Illustration of four-point ordering characterizations of several graph search paradigms.
Solid black edges indicate required adjacencies, dashed edges required non-adjacencies;
all other adjacencies are irrelevant.
The brace specifies a constrained positional scope for~$d$ relative to~$a$ and~$b$,
whose exact position depends on the paradigm.
Figure adapted from~\cite{DBLP:journals/siamdm/CorneilK08}.}
    \label{fig-properties}
\end{figure}

\begin{lemma}[\cite{DBLP:journals/siamdm/CorneilK08}]
    \label{lem-four-vertex}
    Let $G$ be a graph and $\sigma$ an ordering of~$\vset{G}$. Then:
      \begin{enumerate}[(1)]
        \item\label{four-point-gs} $\sigma$ is a GS-ordering of $G$ if and only if $\sigma$ satisfies property S;
        \item\label{four-point-bfs} $\sigma$ is a BFS-ordering of $G$ if and only if $\sigma$ satisfies property B;
        \item\label{four-point-dfs} $\sigma$ is a DFS-ordering of $G$ if and only if $\sigma$ satisfies property D;
        \item\label{four-point-lexbfs} $\sigma$ is a LexBFS-ordering of $G$ if and only if $\sigma$ satisfies property LB;
        \item\label{four-point-lexdfs} $\sigma$ is a LexDFS-ordering of $G$ if and only if $\sigma$ satisfies property LD;
        \item\label{four-point-mns} $\sigma$ is an MNS-ordering of $G$ if and only if $\sigma$ satisfies property M.
    \end{enumerate}
\end{lemma}

We note that no analogous four-point ordering characterization exists for MCS.

\begin{lemma}[\cite{DBLP:journals/siamdm/CorneilK08}]
\label{lem-graph-traversal-relation}
For every graph $G$, the following relations hold among graph search paradigms:
\begin{itemize}
    \item every BFS-, DFS-, and MNS-ordering of $G$ is a GS-ordering of $G$;
    \item every LexBFS-ordering of $G$ is both a BFS-ordering and an MNS-ordering of $G$;
    \item every LexDFS-ordering of $G$ is both a DFS-ordering and an MNS-ordering of $G$;
    \item every MCS-ordering of $G$ is an MNS-ordering of $G$.
\end{itemize}
\end{lemma}

Figure~\ref{fig-relation-graph-searches} illustrates these containment relationships.

\begin{figure}
\begin{minipage}{0.5\textwidth}
     \centering
    \includegraphics[width=0.5\linewidth]{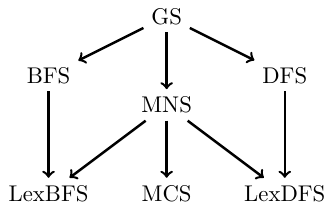}
\end{minipage}\begin{minipage}{0.5\textwidth}
    \caption{Containment relationships among graph search paradigms.
An arc from a search~$\mathcal{S}$ to~$\mathcal{S}'$ indicates that, for any graph~$G$,
every $\mathcal{S}'$-ordering of~$G$ is also an $\mathcal{S}$-ordering of~$G$.
Figure adapted from~\cite{DBLP:journals/siamdm/CorneilK08}.}
    \label{fig-relation-graph-searches}
    \end{minipage}
\end{figure}


\begin{lemma}
\label{lem-bfs-first}
Let~$\sigma$ be a BFS-ordering (respectively, LexBFS-ordering) of a graph~$G$.
Let~$a$ be the first vertex in~$\sigma$, and let~$c$ be a neighbor of $a$ in~$G$.
Then, all vertices between~$a$ and~$c$ in~$\sigma$ are adjacent to~$a$.
\end{lemma}

\begin{proof}
If~$c$ is the second vertex in~$\sigma$, the lemma holds trivially. Otherwise, let~$b$ be a vertex between~$a$ and~$c$ in~$\sigma$. Suppose for contradiction that~$a$ is not adjacent to~$b$ in~$G$. Then, by the four-point ordering characterization of BFS, that is, Statement~(2) of Lemma~\ref{lem-four-vertex} (respectively, for LexBFS, Statement~(4) of Lemma~\ref{lem-four-vertex}), there exists a vertex preceding~$a$ in~$\sigma$ that is adjacent to~$b$. This contradicts the fact that~$a$ is the first vertex in~$\sigma$.
\end{proof}

The following observation is straightforward.

\begin{observation}
\label{ob-connected}
Let~$\mathcal{S}$ be one of the graph traversals: BFS, DFS, LexBFS, LexDFS, MCS, or MNS. Let~$G$ be a connected graph on~$n$ vertices, and let~$\sigma$ be an $\mathcal{S}$-ordering of~$G$. Then, for any $i \in [n]$, the subgraph of~$G$ induced by the first~$i$ vertices in~$\sigma$ is connected.
\end{observation}

\subsection{Problem Definitions}
Let $\mathcal{S}$ be a graph search paradigm.
We say that a graph~$G$ is an \emph{$\mathcal{S}$-support} for a set~$\Pi$ of orderings, or that~$G$ \emph{admits $\Pi$ as a set of $\mathcal{S}$-orderings}, if every ordering in~$\Pi$ is an $\mathcal{S}$-ordering of~$G$. An \emph{$\mathcal{S}$-tree support} is an $\mathcal{S}$-support that is a tree.

We study the following decision problems, originally introduced by Escoffier, Spanjaard, and Tydrichov{\'a}~\cite{DBLP:journals/dam/EscoffierST24} for the special case of GS, although GS is not explicitly mentioned in their paper.

\EPP
{Edge-Bounded $\mathcal{S}$-Single-Peakedness (Edge-Bounded-$\mathcal{S}$-SP)}
{A vertex set~$V$, a set~$\Pi$ of orderings of~$V$, and an integer~$k$.}
{Does $\Pi$ admit an $\mathcal{S}$-support with at most $k$ edges?}

\EPP
{Degree-Bounded $\mathcal{S}$-Single-Peakedness (Deg-Bounded-$\mathcal{S}$-SP)}
{A vertex set $V$, a set $\Pi$ of orderings of $V$, and an integer $k$.}
{Does $\Pi$ admit an $\mathcal{S}$-support with maximum degree at most $k$?}

\EPP
{$\mathcal{S}$-Single-Peakedness  on Trees ($\mathcal{S}$-SP-Trees)}
{A vertex set $V$, and a set $\Pi$ of orderings of $V$.}
{Does~$\Pi$ admit an $\mathcal{S}$-tree support?}


%

\section{Single-Peakedness on Graphs with Bounded Degree or Edge Count}

This section establishes the {\npcns} of \prob{Edge-Bounded-$\mathcal{S}$-SP} and \prob{Deg-Bounded-$\mathcal{S}$-SP} for graph search paradigms not considered in~\cite{DBLP:journals/dam/EscoffierST24}. Since all problems considered below are in {\np}, we focus exclusively on the hardness reductions and omit explicit proofs of {\np}-membership.

Our hardness results are based on reductions from the {\prob{Vertex Cover}} problem.
Given a graph $G=(V,E)$, a subset $S\subseteq V$ is said to \emph{cover} an edge $\edge{u}{v}\in E$ if $\edge{u}{v}\cap S\neq \emptyset$.
A \emph{vertex cover} of~$G$ is a subset $S\subseteq V$ that covers all edges of~$G$.
A \emph{clique} in~$G$ is a set of pairwise adjacent vertices.

\EPP{Vertex Cover}{A graph $G$ and an integer $\kappa$.}{Does $G$ have a vertex cover of $\kappa$ vertices?}

\prob{Vertex Cover} is {\npc} even when restricted to $3$-regular graphs~\cite{DBLP:journals/tcs/GareyJS76}, which we adopt throughout for clarity and consistency, although it is not strictly necessary for all proofs.

The following structure is frequently used in our reductions.

\begin{definition}
    For positive integers~$p$ and~$q$, a~$(p, q)$-drone is a graph of $p+p\cdot q$ vertices and $\frac{p\cdot (p-1)}{2}+p\cdot q$ edges such that:
\begin{enumerate}
    \item[(1)] the graph contains a clique of size~$p$; and
    \item[(2)] every vertex in the clique has exactly~$q$ pendants.
\end{enumerate}
\end{definition}

An example is shown in Figure~\ref{fig-drone}.

\begin{figure}[ht]
    \centering{\includegraphics[width=0.25\textwidth]{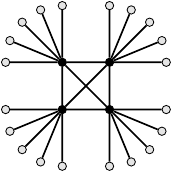}}
   \caption{A $(4,5)$-drone.
   }
    \label{fig-drone}
\end{figure}
\medskip

\noindent{\bf{Outline of the Reductions.}} In several of our reductions, we shall construct a set of orderings that can be classified into two groups. The first group is designed to ensure that any graph with the minimum number of edges (respectively, maximum degree) that {\term{admits}} these orderings as $\mathcal{S}$-orderings contain all edges in $(p, q)$-drones for some specific values of~$p$ and~$q$. This is typically achieved by generating, for each edge of the corresponding $(p, q)$-drone, an ordering in which the two endpoints of the edge appear as the first and second vertices, respectively.
In some reductions, however, the first group  is constructed more carefully to serve additional technical purposes.
It is worth emphasizing that, although we often define a $(p, q)$-drone in our reductions, the drone itself does not belong to the reduction instance. It is defined solely to facilitate the understanding of the reduction and to aid in the correctness proof.
The second group encodes the solution structure of the underlying \prob{Vertex Cover} instance. Specifically, this group ensures that the input graph admits a vertex cover~$S$ of size~$\kappa$ if and only if every graph with at most~$k$ edges (respectively, maximum degree at most~$k$) that admits the orderings in this group as $\mathcal{S}$-orderings contains all edges between~$S$ and a designated vertex~$z$ (not in the $(p, q)$-drone).
In some reductions, the two groups of orderings are not strictly disjoint, as certain orderings may simultaneously fulfill the roles of both classes.



We now begin our exploration of the hardness results.
For clarity, we establish the following  notations that will be used throughout the proofs.
\medskip

\noindent\textbf{Global Notations and Assumptions.}
In all reductions, we consider instances of \prob{Vertex Cover} of the form $I=(G,\kappa)$, where $G$ is a $3$-regular graph.
We denote by $n$ the number of vertices of~$G$ and by $m$ the number of edges of~$G$; since $G$ is $3$-regular, we have $m=3n/2$.
The vertex set and edge set of~$G$ are denoted by~$V$ and~$E$, respectively.
Moreover, $\overrightarrow{V} = (v_1, v_2, \dots, v_n)$ denotes an arbitrary fixed ordering of~$V$, and $\overrightarrow{E} = (e_1, e_2, \dots, e_m)$ denotes an arbitrary fixed ordering of~$E$.
We assume that $n \ge 4$, $m \ge 6$, and $\kappa \ge 2$.
These conventions will be used throughout without further mention.

For two orderings~$X$ and~$Y$, we use $(X, Y)$ to denote their concatenation.
This notation extends naturally to multiple orderings; for instance, if $X = (x_1, x_2, x_3)$, $Y = (y_1, y_2)$, and $Z = (z_1, z_2)$, then
$(X, Y, Z) = (x_1, x_2, x_3, y_1, y_2, z_1, z_2)$.
\medskip

We now present our first hardness result.

\begin{theorem}
\label{thm-bfs-npc}
\prob{Edge-Bounded-BFS-SP} is {\emph{\nph}}.
\end{theorem}

\begin{proof}
Given an instance $I = (G, \kappa)$ of {\prob{Vertex Cover}}, where $G = (V, E)$ is a $3$-regular graph, we construct an instance of {\prob{Edge-Bounded-BFS-SP}} as follows.

We first define a set~$\Pi$ of $2m \cdot n^2 \cdot (n^2 - 1)$ orderings.
These orderings are constructed so that any graph with the minimum number of edges that admits~$\Pi$ as a set of BFS-orderings must contain all edges of a specific graph~$G'$, obtained from~$G$ by the following operations:
\begin{itemize}
\item[(1)] Add edges to make~$V$ a clique.
\item[(2)] For each $i \in [n]$, add a set $F_i = \{f_i^1, f_i^2, \ldots, f_i^{n^2}\}$ of~$n^2$ pendants of~$v_i$. 
\item[(3)] Add an isolated vertex~$z$.
\end{itemize}
Clearly, $G' - z$ is an $(n, n^2)$-drone, and~$G'$ has $\tfrac{n\cdot (n - 1)}{2} + n^3$ edges.
We now detail the construction of the orderings in~$\Pi$.
For each $i \in [n]$, let $\cramped{\overrightarrow{F_i}} = (f_i^1, f_i^2, \ldots, f_i^{n^2})$, and let $\cramped{\overrightarrow{F}} = (\overrightarrow{F_1}, \overrightarrow{F_2}, \ldots, \overrightarrow{F_n})$.

For every edge $\edge{v_i}{v_j} \in \eset{G}$ such that $i,j\in [n]$ and $i<j$, every $t \in [n] \setminus \{i, j\}$, and every pair $\{p, q\} \subseteq [n^2]$ with $p<q$, we construct the following four orderings of~$\vset{G'}$
:
\begin{align*}
\pi((i, j), t, (p, q)) &= \left(f_i^p, v_i, \overrightarrow{F_i} \ominus \{f_i^p, f_i^q\}, f_i^q, v_j, \overrightarrow{V} \ominus \{v_i, v_j, v_t\}, v_t, z, \overrightarrow{F_j}, \overrightarrow{F} \ominus (F_i \cup F_j \cup F_t), \overrightarrow{F_t}\right), \\
\pi((i, j), t, (q, p)) &= \left(f_i^q, v_i, \overrightarrow{F_i} \ominus \{f_i^q, f_i^p\}, f_i^p, v_j, \overrightarrow{V} \ominus \{v_i, v_j, v_t\}, v_t, z, \overrightarrow{F_j}, \overrightarrow{F} \ominus (F_i \cup F_j \cup F_t), \overrightarrow{F_t}\right), \\
\pi((j, i), t, (p, q)) &= \left(f_j^p, v_j, \overrightarrow{F_j} \ominus \{f_j^p, f_j^q\}, f_j^q, v_i, \overrightarrow{V} \ominus \{v_i, v_j, v_t\}, v_t, z, \overrightarrow{F_i}, \overrightarrow{F} \ominus (F_i \cup F_j \cup F_t), \overrightarrow{F_t}\right), \\
\pi((j, i), t, (q, p)) &= \left(f_j^q, v_j, \overrightarrow{F_j} \ominus \{f_j^q, f_j^p\}, f_j^p, v_i, \overrightarrow{V} \ominus \{v_i, v_j, v_t\}, v_t, z, \overrightarrow{F_i}, \overrightarrow{F} \ominus (F_i \cup F_j \cup F_t), \overrightarrow{F_t}\right).
\end{align*}
Let $k = \kappa + \abs{\eset{G'}} = \kappa + \frac{n\cdot (n - 1)}{2} + n^3$. We define the resulting instance of \prob{Edge-Bounded-BFS-SP} as $g(I)=(\vset{G'},\Pi,k)$, which can be constructed in polynomial time. We show that $I$ is a {\yesins} of \prob{Vertex Cover} if and only if $g(I)$ is a {\yesins} of \prob{Edge-Bounded-BFS-SP}.

$(\Rightarrow)$ Assume that~$G$ admits a vertex cover~$S \subseteq V$ of size~$\kappa$.
Let~$H$ be the graph obtained from~$G'$ by adding the~$\kappa$ edges between~$z$ and the vertices in~$S$. Then,~$H$ has exactly $\abs{\eset{G'}} + \kappa = k$ edges.
It is easy to verify that every ordering in~$\Pi$ is a BFS-ordering of~$H$.
As an illustration, let $\sigma$ denote the ordering $\pi((i,j),t,(p,q))$ defined above.
Clearly, $\sbo{\sigma}{v_t}$ is a partial BFS-ordering of~$H$.
The next vertex visited in~$\sigma$ is~$z$.
Since~$S$ is a vertex cover of~$G$, it follows that $\{v_i, v_j\} \cap S \neq \emptyset$.
Consequently, at least one of~$v_i$ and~$v_j$ is adjacent to~$z$ in~$H$.
All vertices between~$v_i$ and~$z$ in~$\sigma$ are neighbors of~$v_i$ in~$H$.
If~$z$ is a neighbor of~$v_i$ in~$H$, then $\sbo{\sigma}{z}$ remains a partial BFS-ordering of~$H$.
Otherwise,~$z$ is a neighbor of~$v_j$. Observe that no vertex succeeding~$z$ in~$\sigma$ is adjacent to~$v_i$, and that every vertex between~$v_j$ and~$z$ in~$\sigma$ is adjacent to~$v_j$ in~$H$.
 It follows that $\sbo{\sigma}{z}$ is again a partial BFS-ordering of~$H$.
The remaining vertices are visited in
$(\overrightarrow{F_j}, \overrightarrow{F} \ominus (F_i \cup F_j \cup F_t), \overrightarrow{F_t})$.
Since every vertex of $F_x$ is adjacent only to $v_x$ for each $x \in [n]$,~$\sigma$ is a BFS-ordering of~$H$.
It follows that $g(I)$ is a {\yesins} of {\prob{Edge-Bounded-BFS-SP}}.

$(\Leftarrow)$ Assume that $g(I)$ is a {\yesins}. Let~$H$ denote now a BFS-support of~$\Pi$ with the minimum number of edges. Our proof breaks down into the following claims.

\begin{claim}
\label{claim-p-q-first-last}
Let $i \in [n]$ and let $\{p, q\} \subseteq [n^2]$ with $p < q$. Then~$\Pi$
\begin{enumerate}[(1)]
    \item contains at least one ordering in which~$f_i^p$ is the first vertex, and~$f_i^q$ succeeds all vertices in~$F_i$, and
    \item contains at least one ordering in which~$f_i^q$ is the first vertex, and~$f_i^p$ succeeds all vertices in~$F_i$.
\end{enumerate}
\end{claim}

\begin{proof}
Let~$v_j$ be an arbitrary neighbor of~$v_i$ in~$G$, where $j\in[n]\setminus\{i\}$, and let~$t$ be an arbitrary element of~$[n]\setminus\{i,j\}$. Since $G$ is $3$-regular, $v_j$ and~$t$ exist.
By construction, the ordering $\pi((i, j), t, (p, q)) \in \Pi$ satisfies the required condition in Statement~(1), and the ordering $\pi((i, j), t, (q, p)) \in \Pi$ satisfies the required condition in Statement~(2) follows analogously.
\end{proof}

\begin{claim}
\label{claim-fi-clique}
Let $i \in [n]$. Then,~$F_i$ induces either a clique or an independent set in~$H$.
\end{claim}

\begin{proof}
If~$F_i$ induces an independent set in~$H$, the claim holds.
Otherwise, let $f_i^p,f_i^q\in F_i$ be adjacent in~$H$.
By Claim~\ref{claim-p-q-first-last}, there exists an ordering in~$\Pi$ where~$f_i^p$ is the first vertex,~$f_i^q$ is the last among vertices of~$F_i$, and every other vertex of~$F_i$ lies between them.
 Then, by Lemma~\ref{lem-bfs-first},~$f_i^p$ is adjacent to all other vertices in~$F_i$ in~$H$.
Now let $f_i^r\in F_i$ be arbitrary.
By Claim~\ref{claim-p-q-first-last} and Lemma~\ref{lem-bfs-first} once more,~$f_i^r$ is adjacent to all vertices in $F_i\setminus \{f_i^r\}$.
Since~$f_i^r$ was arbitrary,~$F_i$ induces a clique in~$H$.
%
%
\end{proof}

\begin{claim}
\label{claim-ub-edges-H}
$H$ contains at most ${n\cdot (n-2)}/{2} + n^3 + n$ edges.
\end{claim}

\begin{proof}
Since $\vset{G}$ is a vertex cover of~$G$, the correctness proof for the $(\Rightarrow)$ direction implies that the graph obtained from~$G'$ by adding edges between~$z$ and~$\vset{G}$ is a BFS-support of~$\Pi$.
 Consequently,
 $\abs{\eset{H}} \leq \abs{\eset{G'}} + n = {n\cdot (n-2)}/{2} + n^3 + n$.
\end{proof}

\begin{claim}
\label{claim-pendant}
Let $i \in [n]$ and let $p \in [n^2]$. Then,~$f_i^p$ is a pendant of~$v_i$ in~$H$.
\end{claim}

\begin{proof}
Suppose for contradiction that~$f_i^p$ has a neighbor~$x$ in~$H$ with $x \neq v_i$.
We distinguish between two exhaustive cases. 

\begin{description}
    \item[Case 1:] $x \in F_i$. \hfill \\
    By Claim~\ref{claim-fi-clique}, the set~$F_i$ induces a clique in~$H$, contributing ${n^2\cdot (n^2 - 1)}/{2}$ edges within~$F_i$. However, since $n \ge 4$, this contradicts Claim~\ref{claim-ub-edges-H}.

    \item[Case 2:] $x \notin F_i$. \hfill \\
    By construction, there exists an ordering in~$\Pi$ in which~$f_i^p$ is the first vertex, and  all vertices of~$F_i$ precede~$x$.
    By Lemma~\ref{lem-bfs-first},~$f_i^p$ is adjacent to every vertex in $F_i \setminus \{f_i^p\}$ in~$H$.
    This reduces to Case~1 and yields a contradiction.
\end{description}

Therefore, $f_i^p$ has no neighbor in~$H$ other than~$v_i$, and is thus a pendant of~$v_i$.
\end{proof}


\begin{claim}
\label{claim-EG-subseteq-EH}
$\eset{G} \subseteq \eset{H}$.
\end{claim}

\begin{proof}
Let $\{v_i, v_j\}$ be an edge in~$G$ with $\{i, j\} \subseteq [n]$ and $i < j$.
Choose any arbitrary $t\in [n] \setminus \{i, j\}$, any pair $\{p, q\}\subseteq [n^2]$ with $p<q$, and consider the ordering $\pi((i, j), t, (p, q))\in \Pi$. In this ordering, the first $n^2 + 1$ vertices are exactly those in $F_i \cup \{v_i\}$, and the next vertex is~$v_j$. By Claim~\ref{claim-pendant}, each vertex in~$F_i$ is adjacent only to~$v_i$ in~$H$. Since this ordering is a BFS-ordering of~$H$, Observation~\ref{ob-connected} implies that~$v_j$ is adjacent to~$v_i$ in~$H$.
As the choice of~$\edge{v_i}{v_j}$ was arbitrary, it follows that
~$\eset{G} \subseteq \eset{H}$.
\end{proof}
	
\begin{claim}
\label{claim-b}
If a vertex~$v_i\in V$ is adjacent to two distinct vertices $v_j, v_t \in V \setminus \{v_i\}$ in~$H$, where $\{i, j, t\} \subseteq [n]$, then~$v_i$ is adjacent to all vertices in $V \setminus \{v_i\}$ in~$H$.
\end{claim}

\begin{proof}
Let~$f_i^p$, $p\in [n^2]$, be an arbitrary vertex from~$F_i$, and let~$v'$ be an arbitrary vertex from $V \setminus \{v_i, v_j, v_t\}$. These vertices exist since $n\geq 4$.
By construction, there exists an ordering $\sigma \in \Pi$ such that:
\begin{enumerate}
    \item[(i)] $f_i^p$ and~$v_i$ are the first and second vertices, respectively; and
    \item[(ii)] 
    $v_i <_{\sigma} v' <_{\sigma} v_t$.
\end{enumerate}
Assume for contradiction that~$v_i$ is not adjacent to~$v'$ in~$H$.
By the four-point ordering characterization of BFS (Lemma~\ref{lem-four-vertex}, Statement~\ref{four-point-bfs}) and Condition~(ii), a vertex preceding~$v_i$ in~$\sigma$ must be adjacent to~$v'$ in~$H$.
By Condition~(i), the only vertex preceding~$v_i$ is~$f_i^p$.
By Claim~\ref{claim-pendant},~$f_i^p$ is a pendant of~$v_i$ in~$H$, a contradiction.
Since $v'$ was chosen arbitrarily, $v_i$ is adjacent to all vertices in $V\setminus\{v_i\}$ in~$H$.
\end{proof}

Since~$G$ is $3$-regular, Claims~\ref{claim-EG-subseteq-EH} and~\ref{claim-b} together imply that~$V$ induces a clique in~$H$.
By Claim~\ref{claim-pendant}, for every $i \in [n]$, all vertices of~$F_i$ are pendants of~$v_i$ in~$H$.
Hence,~$H - z$ is an $(n, n^2)$-drone.
Equivalently,~$H - z$ and~$G' - z$ are identical.
Since~$z$ is isolated in~$G'$, it follows that $\eset{G'} = \eset{H - z}$.

We now show that if~$g(I)$ is a {\yesins} of \prob{Edge-Bounded-BFS-SP}, then~$G$ admits a vertex cover of size~$\kappa$.
To this end, suppose that $\abs{\eset{H}} \leq k = \kappa + \abs{\eset{G'}}$. Define $\kappa' = \abs{\eset{H}} - \abs{\eset{G'}}$, so that $\kappa' \leq \kappa$.
Let $S = \neighbor{z}{H}$. 
By Claim~\ref{claim-pendant}, we have $S \subseteq V$.
Since $H-z = G'-z$, the only edges of~$H$ not contained in~$G'$ are those incident to~$z$.
Hence, $\abs{S} = \abs{\eset{H}} - \abs{\eset{G'}} = \kappa' \leq \kappa$.

\begin{claim}
\label{claim-s-vertex-cover}
$S$ is a vertex cover of~$G$.
\end{claim}

\begin{proof}
Let $\edge{v_i}{v_j}$ be an arbitrary edge in~$G$ with $i,j\in [n]$ and $i < j$.
Choose any arbitrary pair $\{p, q\} \subseteq [n^2]$ with $p < q$, and any $t \in [n] \setminus \{i, j\}$, both of which exist since $n \geq 4$.
Consider the ordering
\[
\sigma =\pi((i, j), t, (p, q))= \left(f_i^p, v_i, \overrightarrow{F_i} \ominus \{p, q\}, f_i^q, v_j, \overrightarrow{V} \ominus \{v_i, v_j, v_t\}, v_t, z, \overrightarrow{F_j}, \overrightarrow{F} \ominus (F_i \cup F_j \cup F_t), \overrightarrow{F_t} \right)
\]
constructed in~$\Pi$.
If~$z$ is adjacent to~$v_j$ in~$H$, then $v_j \in S$.
Otherwise, let~$f_j$ be any vertex in~$F_j$.
By Claim~\ref{claim-pendant},~$f_j$ is adjacent to~$v_j$ in~$H$. Since $v_j <_{\sigma} z <_{\sigma} f_j$, the four-point ordering  characterization of BFS (Lemma~\ref{lem-four-vertex}, Statement~\ref{four-point-bfs}) implies that some vertex preceding~$v_j$ in~$\sigma$, that is, some vertex in
$F_i \cup \{v_i\}$, is adjacent to~$z$ in~$H$.
By Claim~\ref{claim-pendant}, no vertex from~$F_i$ is adjacent to~$z$ in~$H$, and therefore~$v_i$ must be adjacent to~$z$, implying $v_i \in S$.
Since the choice of the edge~$\edge{v_i}{v_j}$ was arbitrary,~$S$ covers all edges of~$G$.
\end{proof}

By Claim~\ref{claim-s-vertex-cover} and the bound~$\abs{S} \leq \kappa$, we conclude that~$I$ is a {\yesins} of \prob{Vertex Cover}. 
\end{proof}

%

The above proof also applies to LexBFS.

\begin{theorem}
\label{thm-lexbfs-npc}
\prob{Edge-Bounded-LexBFS-SP} is {\emph{\nph}}.
\end{theorem}

\begin{proof}
We adopt the same reduction presented in Theorem~\ref{thm-bfs-npc}.

For the correctness of the $(\Rightarrow)$ direction, we let $S \subseteq \vset{G}$ denote a vertex cover of $G$ of size~$\kappa$, and prove that the graph~$H$ obtained from $G'$ by adding edges between~$z$ and all vertices in~$S$ is a LexBFS-support of~$\Pi$.

\begin{claim}
Every ordering constructed in~$\Pi$ is a LexBFS-ordering of~$H$.
\end{claim}

\begin{proof}
By symmetry, it suffices to show that the ordering $\pi((i, j), t, (p, q))$ defined in the proof of Theorem~\ref{thm-bfs-npc}, is a LexBFS ordering of~$H$. Let~$\sigma$ denote this ordering.

The first vertex in~$\sigma$ is~$f_i^p$, followed by its unique neighbor~$v_i$ in $H$. Thus, $\sigma_{\leq 2}$ forms a partial LexBFS ordering of~$H$. Next,~$\sigma$ visits all remaining vertices in $F_i$, which form an independent set in~$H$; hence $\sigma_{\leq n^2 + 1}$ remains a partial LexBFS-ordering of~$H$.

Subsequently, the ordering proceeds through the remaining $n - 1$ vertices in $V$, which, together with $v_i$, form a clique in $H$. Note that none of the vertices in $V \setminus \{v_i\}$ is adjacent to any vertex in $F_i$. Consequently, $\sigma_{\leq n^2 + n}$ is a partial LexBFS-ordering of~$H$.

The next vertex visited by~$\sigma$ is~$z$.
Since~$S$ is a vertex cover of~$G$, we have $\{v_i,v_j\}\cap S\neq\emptyset$,
and thus~$z$ is adjacent to at least one of $v_i$ and~$v_j$ in~$H$.
Accordingly,~$z$ has distance at most three from~$f_i^p$.
All vertices visited after~$z$ are in $\overrightarrow{F}\ominus F_i$ and have distance three from~$f_i^p$.
Moreover, no vertex in this set is adjacent in~$H$ to any vertex preceding~$v_j$ in~$\sigma$,
including~$v_i$.
It follows that $\sigma_{\le n^2+n+1}$ remains a partial LexBFS-ordering of~$H$.


The remaining suffix of~$\sigma$ is
$(\overrightarrow{F_j},\, \overrightarrow{F}\ominus(F_i\cup F_j\cup F_t),\, \overrightarrow{F_t})$, for which
the following properties hold:
\begin{itemize}
    \item For every $x\in[n]\setminus\{i\}$, the vertices of~$F_x$ appear consecutively in~$\sigma$.
    \item For any $x,y\in[n]\setminus\{i\}$, if $v_x<_\sigma v_y$, then $F_x<_\sigma F_y$.
\end{itemize}
Since for every $x\in[n]$ each vertex of~$F_x$ is a pendant of~$v_x$ in~$H$,
these properties imply that~$\sigma$ is a LexBFS-ordering of~$H$.
\end{proof}

For the $(\Leftarrow)$ direction, note that every LexBFS-ordering is also a BFS-ordering of the same graph
(Lemma~\ref{lem-graph-traversal-relation}).
Therefore, the $(\Leftarrow)$ direction follows directly from the corresponding argument
in the proof of Theorem~\ref{thm-bfs-npc}.
\end{proof}

The reduction for \prob{Edge-Bounded-BFS-SP} can be adapted to show the hardness of the variant
that bounds the maximum degree.

\begin{theorem}
\label{thm-bfs-npc-delta}
\prob{Deg-Bounded-BFS-SP} is {\emph{\nph}}.
\end{theorem}

\begin{proof}
We adapt the reduction in the proof of Theorem~\ref{thm-bfs-npc}.
Starting from the graph~$G'$ constructed there, we add a set~$Z$ of $n\cdot (n+1)$ pendants of~$z$.
The resulting graph~$G'$ has $n + n^2 + n \cdot (n+1) = 2n + 2n^2$
vertices and
$
{n\cdot (n-1)}/{2} + n^3 + n \cdot (n+1)
$
edges.
Moreover, the subgraph $G' - (Z \cup \{z\})$ forms an $(n, n^2)$-drone.
We now define the family~$\Pi$ of orderings.
Let $\overrightarrow{Z} = (z_1, z_2, \ldots, z_{n \cdot (n+1)})$
be an arbitrary ordering of~$Z$.

For every edge $\edge{v_i}{v_j}\in E(G)$ with $i<j$,
every pair $\{p,q\}\subseteq[n^2]$ with $p<q$,
every $t\in[n]\setminus\{i,j\}$,
and every pair $\{\ell,r\}\subseteq[n\cdot (n+1)]$ with $\ell<r$,
we add the following orderings to~$\Pi$.

Let $\alpha\in\{(i,j),(j,i)\}$ and $\beta\in\{(p,q),(q,p)\}$.
For each choice of~$\alpha$ and~$\beta$, we include
\begin{align*}
\pi'(\alpha,t,\beta, (\ell, r))=\bigl(\pi(\alpha,t,\beta),\, z_\ell,\, \overrightarrow{Z}\setminus\{z_\ell,z_r\},\, z_r\bigr),\\
\pi'(\alpha,t,\beta, (r, \ell))=\bigl(\pi(\alpha,t,\beta),\, z_r,\, \overrightarrow{Z}\setminus\{z_\ell,z_r\},\, z_{\ell}\bigr).
\end{align*}
Here, $\pi(\alpha,t,\beta)$ is defined as in the proof of
Theorem~\ref{thm-bfs-npc}.

Let $k=\kappa+n\cdot (n+1)$ and define $g(I)=(\vset{G'},\Pi,k)$
as an instance of \prob{Deg-Bounded-BFS-SP}.
This instance can be constructed in polynomial time.
We now prove the correctness of the reduction.

$(\Rightarrow)$ Assume that~$G$ admits a vertex cover~$S$ of size~$\kappa$.
Let~$H$ be the graph obtained from~$G'$ by adding the edges between~$z$ and the vertices in~$S$.
In~$H$, the vertex~$z$ has the maximum degree $\kappa + n\cdot (n+1)=k$.
By an argument analogous to that in the proof of Theorem~\ref{thm-bfs-npc},~$H$ is a BFS-support of~$\Pi$.

$(\Leftarrow)$ Assume that~$g(I)$ is a {\yesins} of \prob{Deg-Bounded-BFS-SP}; that is,~$\Pi$ admits a BFS-support~$H$ with maximum degree at most~$k$.
To show that~$G$ has a vertex cover of size~$\kappa$, we establish the following structural properties of~$H$.

\begin{claim}
\label{claim-f-v-connected-degree-bfs}
Let $i \in [n]$ and $p\in [n^2]$. Then,~$f_i^p$ is adjacent to~$v_i$ in~$H$.
\end{claim}

\begin{proof}
By construction, there exists an ordering in~$\Pi$ in which~$f_i^p$ and~$v_i$ are the first and second vertices, respectively.
Since~$H$ is a BFS-support of~$\Pi$,~$f_i^p$ and~$v_i$ are adjacent in~$H$ (Observation~\ref{ob-connected}).
\end{proof}

Let $F=\bigcup_{b\in [n]}F_b$.

\begin{claim}
\label{claim-lexbfs-npc-delta-4}
Let $i \in [n]$ and $p \in [n^2]$. Then, $\neighbor{f_i^p}{H} \subseteq F_i \cup \{v_i\}$.
\end{claim}

\begin{proof}
Suppose, for contradiction, that $f_i^p \in F_i$ has a neighbor $x \notin F_i \cup \{v_i\}$ in~$H$.
We distinguish three exhaustive cases according to the location of~$x$.

\begin{description}
    \item[Case 1:] $x \in F \setminus F_i$.\hfill

    Let $t \in [n] \setminus \{i\}$ be such that $x \in F_t$.
    Let~$v_j$ be any neighbor of~$v_i$ in~$G$, where $j \in [n] \setminus \{i\}$, which exists since~$G$ is $3$-regular.
    Choose an arbitrary $q \in [n^2] \setminus \{p\}$, and any pair $\{\ell, r\} \subseteq [n\cdot (n+1)]$ with $\ell < r$.
    By construction, there exists an ordering~$\sigma\in \Pi$ corresponding to the edge $\{v_i, v_j\}$, the pair $\{p, q\}$, the integer~$t$, and the pair $\{\ell, r\}$ such that:
    \begin{itemize}
        \item $f_i^p$ appears first, and
        \item all vertices of~$F_t$, including~$x$, succeed all $n + n^2\cdot (n-1)$ vertices in $V \cup (F \setminus F_t)$.
    \end{itemize}
    Since~$\sigma$ is a BFS-ordering of~$H$, Lemma~\ref{lem-bfs-first} implies that~$f_i^p$ is adjacent to all vertices preceding~$x$ in~$\sigma$.
    Therefore, the degree of~$f_i^p$ in~$H$ is at least
    \[
    \abs{V} + \abs{F \setminus F_t} + 1 = n + n^2\cdot (n-1) + 1 > k = \kappa + n\cdot (n+1),
    \]
    provided $n \geq \kappa \geq 4$.
    This contradicts the assumption that~$H$ has maximum degree at most~$k$.

    \item[Case 2:] $x \in V \setminus \{v_i\}$.\hfill

    Let $x = v_j$ for some $j \in [n] \setminus \{i\}$.
    As in Case~1, there exists an ordering $\sigma \in \Pi$ such that:
    \begin{itemize}
        \item some $f_j^{p'} \in F_j$, where $p'\in [n^2]$, appears first,
        \item $v_j$ appears second, and
        \item $f_i^p$ succeeds all~$n^2\cdot (n-1)$ vertices in $F \setminus F_i$.
    \end{itemize}
    By the argument in Case~1,~$f_j^{p'}$ is not adjacent to any vertex in $F \setminus F_j$, in particular not to~$f_i^p$.
    Since~$v_j$ is adjacent to~$f_i^p$, and~$\sigma$ is a BFS-ordering of~$H$, the four-point ordering characterization of BFS (Lemma~\ref{lem-four-vertex}, Statement~\ref{four-point-bfs}) implies that~$v_j$ is adjacent to all vertices in $F \setminus (F_i \cup F_j)$ between~$v_j$ and~$f_i^p$ in~$\sigma$.
    Therefore,~$v_j$ has at least $1 + n^2\cdot (n-2)$ neighbors in~$H$.
    Since $n \geq \kappa \geq 4$, this again contradicts that~$H$ has maximum degree at most~$k$.

    \item[Case 3:] $x \in Z \cup \{z\}$.\hfill

    As in Case~1, there exists an ordering in~$\Pi$ where~$f_i^p$ appears first, and all vertices in~$V$ precede all vertices in $Z \cup \{z\}$, including~$x$.
    Since~$f_i^p$ is adjacent to~$x$ in~$H$, and the ordering is a BFS-ordering of~$H$, the four-point ordering characterization of BFS (Lemma~\ref{lem-bfs-first}) implies that~$f_i^p$ is adjacent to all vertices in~$V$.
    This reduces to Case~2, which has already been shown to yield a contradiction.
\end{description}

We conclude that $\neighbor{f_i^p}{H} \subseteq F_i \cup \{v_i\}$.
\end{proof}

\begin{claim}
\label{claim-lexbfs-npc-delta-5}
    No vertex in~$Z$ is adjacent to any vertex from~$V$ in~$H$.
\end{claim}

 \begin{proof}
Assume for contradiction that there exists $z' \in Z$ adjacent to some $v_i \in V$, where $i \in [n]$.
Let~$f$ be an arbitrary vertex from~$F_i$.
There exists an ordering~$\sigma \in \Pi$ in which~$f$,~$v_i$, and~$z'$ appear first, second, and last, respectively.
By Claim~\ref{claim-lexbfs-npc-delta-4},~$f$ is not adjacent to any vertex from $F \setminus F_i$ in~$H$.
Since~$\sigma$ is a BFS-ordering of~$H$, the four-point ordering characterization of BFS (Lemma~\ref{lem-four-vertex}, Statement~2) implies that~$v_i$ is adjacent to all vertices in $F \setminus F_i$, and hence has degree at least $1 + n^2 \cdot (n - 1)$ in~$H$.
However, since $n\geq \kappa\geq 4$, this contradicts the assumption that~$H$ has maximum degree at most $\kappa + n \cdot (n+1)$.
\end{proof}

\begin{claim}
\label{claim-lexbfs-npc-delta-6}
Every vertex in~$Z$ is adjacent to~$z$ in~$H$.
\end{claim}

\begin{proof}
Let $z' \in Z$.
There exists an ordering $\sigma \in \Pi$ in which the set of predecessors of~$z'$ is
exactly $F \cup V \cup \{z\}$.
Since~$\sigma$ is a BFS-ordering of~$H$,~$z'$ is adjacent to at least one of its predecessors (Observation~\ref{ob-connected}), and by
Claims~\ref{claim-lexbfs-npc-delta-4} and~\ref{claim-lexbfs-npc-delta-5}, this predecessor must be~$z$.
\end{proof}

Let $S = \neighbor{z}{H} \cap V$ be the set of neighbors of~$z$ in~$V$ within the graph~$H$.

\begin{claim}
$S$ is a vertex cover of~$G$.
\end{claim}

\begin{proof}
Let $\edge{v_i}{v_j}$ be an arbitrary edge of~$G$, where $i,j\in[n]$ and $i<j$.
Choose arbitrary $\{p,q\}\subseteq[n^2]$ with $p<q$ and any
$t\in[n]\setminus\{i,j\}$.
In the current reduction, there exists an ordering in~$\Pi$ whose prefix is
$\pi((i,j),t,(p,q))$ defined in the proof of Theorem~\ref{thm-bfs-npc}.
The remainder of the argument is identical to the proof of
Claim~\ref{claim-s-vertex-cover} in the proof of
Theorem~\ref{thm-bfs-npc}, which shows that $\{v_i, v_j\}\cap S\neq\emptyset$.
Since the choice of $\edge{v_i}{v_j}$ was arbitrary, $S$ is a vertex cover of~$G$.
\end{proof}

By Claim~\ref{claim-lexbfs-npc-delta-6} and the fact that~$H$ has maximum degree at
most~$k$, the vertex~$z$ has at most $k - \abs{Z} = \kappa$ neighbors in~$V$.
Thus, $\abs{S} \le \kappa$.
Any $\kappa$-subset $S' \subseteq V$ with $S \subseteq S'$ is a
{\yes}-witness for the instance~$I$.
\end{proof}

Using the same reduction, the NP-hardness proof extends to LexBFS.

\begin{theorem}
\label{thm-lexbfs-npc-delta}
\prob{Deg-Bounded-LexBFS-SP} is {\emph{\nph}}.
\end{theorem}

\begin{proof}
We use the same reduction as in the proof of Theorem~\ref{thm-bfs-npc-delta}.
For the $(\Rightarrow)$ direction, it was shown there that if~$G$ has a vertex cover~$S$ of size $\kappa$, then the graph $H$, obtained from $G'$ by adding the $\kappa$ edges between $z$ and the vertices in $S$, has maximum degree $k$ and is a BFS-support of~$\Pi$. It therefore suffices to show that~$H$ is also a LexBFS-support of~$\Pi$.
Consider an ordering $\sigma=\pi'((i, j), t, (p, q), (\ell, r))$. By construction, $\sigma$ has $\pi((i, j), t, (p, q))$ as a prefix. By construction, the subgraph of~$H$ induced by the vertices in this prefix is identical to that used in the proof of Theorem~\ref{thm-lexbfs-npc}, which, as we have shown, admits~$\Pi$ as a set of LexBFS-orderings.
The remaining vertices in~$\sigma$ are those of~$Z$, which form an
independent set and are pendants of~$z$ in~$H$. Hence,~$\sigma$ is a LexBFS-ordering of~$H$.
The correctness of the $(\Leftarrow)$ direction follows analogously to the
proof of Theorem~\ref{thm-bfs-npc-delta}.
\end{proof}

We next establish NP-hardness for DFS.

\begin{theorem}
\label{thm-dfs-npc}
\prob{Edge-Bounded-DFS-SP} is {\emph{\nph}}.
\end{theorem}

\begin{proof}
Given an instance $I = (G, \kappa)$ of {\prob{Vertex Cover}}, where $G = (V, E)$ is a $3$-regular graph, we construct an instance of \prob{Edge-Bounded-DFS-SP} as follows.
Let~$G'$ be the graph obtained from~$G$ by:
\begin{itemize}
	\item[(1)] adding edges to make~$V$ a clique;
	
	\item[(2)] for every $i \in [n]$, adding a pendant vertex~$f_i$ adjacent to~$v_i$;
	
	\item[(3)] adding an isolated vertex~$z$.
\end{itemize}
The graph~$G'$ has $2n+1$ vertices and ${n\cdot(n-1)}/{2}+n$ edges, and $G'-z$ is an $(n, 1)$-drone.
Let $\overrightarrow{VF} = (v_1, f_1, v_2, f_2, \ldots, v_n, f_n)$.
We construct a set~$\Pi$ of orderings as follows:
\begin{enumerate}[(i)]
	\item For each pair $\{i, j\} \subseteq [n]$ with $i < j$, add to~$\Pi$ the ordering
	\[\sigma_{(i<j)}=(v_i, v_j, f_j, \overrightarrow {VF}\ominus \{v_i, f_i, v_j, f_j\}, z, f_i).\]
	
	\item For each $i \in [n]$, add to~$\Pi$ the ordering
	\[\sigma_{i}=(f_i, v_i, \overrightarrow {VF}\ominus \{v_i, f_i\}, z).\]
	
	\item For each edge $\edge{v_i}{v_j}$ in~$G$, where $\{i,j\}\subseteq [n]$ and $i<j$, add to~$\Pi$ the following two orderings:
	\begin{align*}
    \sigma_{(i,\, j)}&=(\overrightarrow {VF}\ominus \{v_i, f_i, v_j, f_j\}, v_i, v_j, f_j, z, f_i),\\
	\sigma_{(j,\, i)}&=(\overrightarrow {VF}\ominus \{v_j, f_j, v_i, f_i\}, v_j, v_i, f_i, z, f_j).
    \end{align*}
\end{enumerate}
We have $\abs{\Pi}={n\cdot (n-1)}/{2}+n+2m={n\cdot (n-1)}/{2}+4n$.
Let $k=\kappa + \abs{\eset{G'}}$, and define $g(I)=(V, \Pi,k)$.
We prove below that~$I$ is an instance of {\prob{Vertex Cover}} if and only if~$g(I)$ is an {\yesins} of \prob{Edge-Bounded-DFS-SP}.

$(\Rightarrow)$	Assume that~$G$ admits a vertex cover~$S$ of~$\kappa$ vertices.
	Let~$H$ be the graph obtained from~$G'$ by adding the~$\kappa$ edges between~$z$ and the vertices from~$S$.
 It is not difficult to verify that every ordering constructed in~(i)--(ii) is a DFS-ordering of~$H$.
 Consider an ordering~$\sigma_{(i, j)}$ created in~(iii) corresponding to an edge~$\edge{v_i}{v_j}$ with~$i < j$.
The prefix of~$\sigma_{(i, j)}$ preceding~$z$ forms a partial DFS-ordering of~$H$.
After visiting~$f_j$, DFS backtracks to~$v_j$.
If~$z$ is adjacent to~$v_j$ in~$H$, then~$z$ can be visited next, and~$f_i$, which is a pendant of~$v_i$ in~$H$, can be visited last.
Otherwise, neither of the remaining vertices is adjacent to~$v_j$, so the DFS backtracks further to~$v_i$.
Since~$S$ is a vertex cover, by the construction of~$H$,~$z$ is adjacent to~$v_i$, and hence~$z$ can be visited next, followed by~$f_i$ as the last vertex.
We conclude that~$\sigma_{(i, j)}$ is a DFS-ordering of~$H$.
The argument for~$\sigma_{(j, i)}$ is analogous.
Since~$H$ contains at most~$k$ edges,~$g(I)$ is a {\yesins} of \prob{Edge-Bounded-DFS-SP}.

$(\Leftarrow)$ Assume that~$g(I)$ is a {\yesins} of \prob{Edge-Bounded-DFS-SP}, and let~$H$ be a DFS-support of~$\Pi$ with the minimum number of edges. Then,~$H$ contains at most $k = \kappa + \frac{n\cdot (n-1)}{2} + n$ edges.
The structure of the DFS-orderings in groups~(i) and~(ii), particularly the placement of the first two vertices, leads to the following observation.

\begin{observation}
\label{obs-dfg}
    $\eset{G'} \subseteq \eset{H}$.
\end{observation}

We now conclude the proof with the following claims.

\begin{claim}
\label{claim-dfg-aa}
Let $\edge{v_i}{v_j}$ be an edge in~$G$ with $i, j \in [n]$ and $i<j$. Then, in~$H$:
\begin{enumerate}
    \item[(a)] $z$ is adjacent to at least one of~$v_i$,~$v_j$, or~$f_i$, and
    \item[(b)] $z$ is adjacent to at least one of~$v_i$,~$v_j$, or~$f_j$.
\end{enumerate}
\end{claim}

\begin{proof}
If~$z$ is adjacent to at least one of~$v_i$ or~$v_j$, both statements hold.
Hence, suppose that~$z$ is adjacent to neither~$v_i$ nor~$v_j$.

Consider the last five vertices in the ordering~$\sigma_{(i,j)}$ constructed in~(iii), appearing in the relative order $(v_i, v_j, f_j, z, f_i)$.
By Observation~\ref{obs-dfg},~$v_i$ and~$f_i$ are adjacent in~$H$.
Since~$\sigma_{(i,j)}$ is a DFS-ordering of~$H$, Lemma~\ref{lem-four-vertex} (Statement~\ref{four-point-dfs}) implies that if~$z$ is not adjacent to~$v_i$, then $z$ is adjacent to at least one of~$v_j$ or~$f_j$.
By assumption,~$z$ is not adjacent to~$v_j$, and hence $z$ is adjacent to~$f_j$, establishing Statement~(b).

To prove Statement~(a), consider the final five vertices $(v_j, v_i, f_i, z, f_j)$ in the ordering~$\sigma_{(j,i)}$ introduced in~(iii).
By Observation~\ref{obs-dfg},~$v_j$ and~$f_j$ are adjacent in~$H$.
Analogously to the proof of Statement~(b), if~$z$ is adjacent to neither~$v_i$ nor~$v_j$, then~$z$ is adjacent to~$f_i$ in~$H$.
\end{proof}

%

Using Claim~\ref{claim-dfg-aa}, we construct a vertex cover of~$G$ as follows. Let $S_V=\neighbor{z}{H}\cap V$ be the set of neighbors of~$z$ in~$V$.
Initialize~$S_F=\emptyset$. For every edge~$\edge{v_i}{v_j}$ in~$G$ such that~$z$ is adjacent to neither~$v_i$ nor~$v_j$ in~$H$, arbitrarily add one of~$v_i$ or~$v_j$ to~$S_F$.
By construction,~$S_V \cup S_F$ is a vertex cover of~$G$.
We next show that this set contains at most~$\kappa$ vertices.

\begin{claim}
\label{claim-dfs-npc-last}
    $\abs{S_V \cup S_F} \leq \kappa$.
\end{claim}

\begin{proof}
Let~$h=\abs{S_V}$ denote the number of edges between~$z$ and vertices in~$V$ in the graph~$H$, and let~$h'$ denote the number of edges between~$z$ and vertices in~$\{f_1, f_2, \dots, f_n\}$. Thus, $h + h'$ equals the degree of~$z$ in~$H$.

Consider an edge $\edge{v_i}{v_j} \in \eset{G}$. By Claim~\ref{claim-dfg-aa}, if~$z$ is adjacent to neither~$v_i$ nor~$v_j$ in~$H$, then~$z$ is adjacent to both~$f_i$ and~$f_j$. Consequently, for every vertex added to~$S_F$, there exists a distinct neighbor of~$z$ in~$\{f_1,\dots,f_n\}$, implying that~$\abs{S_F}\le h'$.

Since~$H$ contains at most $k = \kappa + \abs{\eset{G'}}$ edges and~$G'$ is a subgraph of~$H$ (Observation~\ref{obs-dfg}), we have:
\[
h + h' \leq \abs{\eset{H}} - \abs{\eset{G'}} \leq k - \abs{\eset{G'}} = \kappa.
\]
Therefore,
\[
\abs{S_V \cup S_F} = \abs{S_V} + \abs{S_F} \leq h + h' \leq \kappa,
\]
proving the claim.
\end{proof}

By the above claim, we conclude that~$S_V \cup S_F$ is a vertex cover of~$G$ of size at most~$\kappa$,
and hence the instance~$I$ is a {\yesins} of \prob{Vertex Cover}.
This completes the proof of Theorem~\ref{thm-dfs-npc}.
\end{proof}

Building upon the reduction presented in the previous proof, we obtain the following result.

\begin{theorem}
\label{thm-dfs-npc-delta}
\prob{Deg-Bounded-DFS-SP} is {\nph}.
\end{theorem}

\begin{proof}
We prove the theorem by modifying the reduction in the proof of Theorem~\ref{thm-dfs-npc}.
Let $I = (G, \kappa)$ be an instance of \prob{Vertex Cover} on $3$-regular graphs.
Starting from the graph~$G'$ constructed there, we add a set~$Z$ of~$n$ pendants of the vertex~$z$. The resulting graph~$G'$ has $2n + 1$ vertices and ${n\cdot (n-1)}/{2} + 2n$ edges, and $G' - (Z \cup \{z\})$ is an $(n,1)$-drone.
Let $\overrightarrow{VF}$ be defined as before, and let $\overrightarrow{Z} = (z_1, z_2, \dots, z_n)$ be an arbitrary ordering of~$Z$.
We now construct a set~$\Pi$ of orderings of~$\vset{G'}$ as follows:
\begin{enumerate}[(i)]
    \item For each $i, j \in [n]$ with $i < j$, include the ordering
    \[
    \sigma_{(i<j)} = \left(v_i, v_j, f_j, \overrightarrow{VF} \ominus \{v_i, f_i, v_j, f_j\}, z, \overrightarrow{Z}, f_i\right).
    \]

    \item For each $i \in [n]$ and each $t \in [n]$, include the ordering
    \[
    \sigma_{i} = \left(f_i, v_i, \overrightarrow{VF} \ominus \{v_i, f_i\}, z, z_t, \overrightarrow{Z} \ominus \{z_t\}\right).
    \]

    \item For each edge $\edge{v_i}{v_j}$ in~$G$ with $i < j$, include the two orderings:
    \[
    \sigma_{(i, j)} = \left(\overrightarrow{VF} \ominus \{v_i, f_i, v_j, f_j\}, v_i, v_j, f_j, z, \overrightarrow{Z}, f_i\right),
    \]
    \[
    \sigma_{(j, i)} = \left(\overrightarrow{VF} \ominus \{v_j, f_j, v_i, f_i\}, v_j, v_i, f_i, z, \overrightarrow{Z}, f_j\right).
    \]
\end{enumerate}

Let $k = \kappa + n$, and define $g(I) = (\vset{G'}, \Pi, k)$ as an instance of \prob{Deg-Bounded-DFS-SP}.
We show that $G$ has a vertex cover of size~$\kappa$ if and only if $\Pi$ admits a DFS-support of maximum degree at most~$k$.

{($\Rightarrow$)} Suppose~$G$ has a vertex cover $S \subseteq \vset{G}$ of size~$\kappa$. Let~$H$ be the graph obtained from~$G'$ by adding the~$\kappa$ edges between~$z$ and the vertices of~$S$. Then~$z$ has degree~$\kappa + n = k$ in~$H$, and every other vertex has degree at most~$n + 2$. It remains to show that all orderings constructed above are DFS-orderings of~$H$. This is immediate for the orderings in~(i) and~(ii).
Now consider the two orderings in~(iii) corresponding to an edge $\edge{v_i}{v_j}\in \eset{G}$ with $i < j$.
For both $\sigma_{(i, j)}$ and $\sigma_{(j, i)}$, the prefix consisting of all vertices preceding~$z$ clearly forms a partial DFS-ordering of~$H$.
By the same argument as in the proof of Theorem~\ref{thm-bfs-npc}, the next visited vertex~$z$ is adjacent to at least one of~$v_i$ or~$v_j$, since~$S$ is a vertex cover of~$G$. Recall also that none of the other unvisited vertices is adjacent to~$v_j$ in~$H$.
Thus, after visiting~$z$, both orderings remain partial DFS-orderings of~$H$.
Each ordering then visits the~$n$ pendant vertices of~$z$, followed by the remaining pendant of~$v_i$ in $\sigma_{(i, j)}$ (or of~$v_j$ in $\sigma_{(j, i)}$).
Hence, both orderings are DFS-orderings of~$H$.

{($\Leftarrow$)} Suppose~$g(I)$ is a {\yesins} of \prob{Deg-Bounded-DFS-SP}, and let~$H$ be a DFS-support of~$\Pi$ with maximum degree at most~$k=\kappa + n$. Observe that Observation~\ref{obs-dfg} and Claim~\ref{claim-dfg-aa} from the proof of Theorem~\ref{thm-dfs-npc} apply unchanged. Let~$S_V$ and~$S_F$ be defined as in the proof of Theorem~\ref{thm-dfs-npc}. By the same argument as in the proof of Claim~\ref{claim-dfs-npc-last}, we conclude that $S_V \cup S_F$ is a vertex cover of~$G$ of size at most~$\kappa$.

\end{proof}

We now turn to \textsc{LexDFS}.
Let~$\mathcal{S}$ be a graph search paradigm studied in the paper.
Given a graph~$G$ and an ordering~$\sigma$ of a subset~$A \subseteq \vset{G}$,
we define~$\mathcal{S}(G; \sigma)$ as follows:
\begin{enumerate}[(1)]
    \item Construct a graph~$G^{\bigwhitestar}$ from~$G$ by turning~$A$ into a clique.
    \item Let~$\sigma'$ be an arbitrary $\mathcal{S}$-ordering of~$G^{\bigwhitestar}$ having~$\sigma$ as a prefix. Specifically, we first visit the vertices of~$\sigma$ in their given order and then proceed with the remaining vertices according to~$\mathcal{S}$. Such an ordering exists since~$A$ induces a clique in~$G^{\bigwhitestar}$. 
    \item Define~$\mathcal{S}(G; \sigma)$ as~$\sigma' \ominus A$, that is, the ordering obtained from~$\sigma'$ by removing all vertices in~$A$.
\end{enumerate}

\begin{lemma}
\label{lem-prefix}
Let~$G$ be a graph, let $A \subseteq \vset{G}$ be a subset of vertices, and let~$\sigma$ be an ordering on~$A$.
Then~$\sigma$ is a partial $\mathcal{S}$-ordering of~$G$ if and only if the concatenation $(\sigma, \mathcal{S}(G; \sigma))$ is an $\mathcal{S}$-ordering of~$G$.
\end{lemma}

\begin{proof}
The ``only if'' direction is immediate.
For the ``if'' direction, let $\pi = (\sigma, \mathcal{S}(G; \sigma))$.
Suppose that~$\sigma$ is a partial $\mathcal{S}$-ordering of~$G$.
To prove that~$\pi$ is an $\mathcal{S}$-ordering of~$G$, it suffices to show that for every vertex $v \in \vset{G} \setminus A$,
the prefix $\sbo{\pi}{v}$ is a partial $\mathcal{S}$-ordering of~$G$.

Let $\mathcal{S}(G; \sigma) = (v_1, v_2, \dots, v_t)$, where $t = \abs{\vset{G}} - \abs{A}$.
By construction, for each $i \in [t]$, when the vertex~$v_i$ is visited during the execution as described above in~(2), its already visited neighbors in~$G^{\bigwhitestar}$ are precisely the neighbors of~$v_i$ in~$G$ that precede~$v_i$ in~$\pi$.
For all graph search paradigms considered in this paper, the selection of the next vertex to visit is determined entirely
by the configuration of visited neighborhoods across all unvisited vertices, namely, by how their sets (or sequences) of visited neighbors compare to one another.
It follows inductively that~$\pi$ is an $\mathcal{S}$-ordering of~$G$.
\end{proof}

We next establish our first hardness result for \textsc{LexDFS}.

\begin{theorem}
\label{thm-lexdfs-npc}
\prob{Edge-Bounded-LexDFS-SP} is {\emph{\nph}}.
\end{theorem}

\begin{proof}
Let $I = (G, \kappa)$ be an instance of \prob{Vertex Cover} on $3$-regular graphs.
We assume $n \geq 10$, ensuring that~$G$ admits a matching of size at least three~\cite{DBLP:journals/dm/BiedlDDFK04,DBLP:journals/tcs/GareyJS76}. Consequently, every minimum vertex cover of~$G$ contains at least three vertices and at most~$n - 1$ vertices. We therefore further assume that $3 \leq \kappa \leq n - 1$.

We now construct an instance~$g(I)$ of \prob{Edge-Bounded-LexDFS-SP} as follows.
We use~$p,q\in[m]$ as subscripts for edges of~$G$, and~$h,i,j,x,y\in[n]$ as subscripts for vertices of~$G$.
Let $G'$ be obtained from $G$ as follows (see Figure~\ref{fig-lexdfs-np-hard}).
\begin{itemize}
    \item[(1)] For each edge $e_p=\edge{v_i}{v_j}$ in~$G$, where $i < j$, introduce three new vertices~$t_p$,~$w_p$, and~$u_p$, and add the following edges:
    \begin{itemize}
        \item $t_p$ is adjacent to both~$v_i$ and~$v_j$;
        \item $w_p$ is adjacent to~$t_p$ and~$v_i$; and
        \item $u_p$ is adjacent to~$t_p$ and to all vertices in~$V$.
    \end{itemize}

    Let $T=\{t_p \setmid p\in [m]\}$, $U=\{u_p \setmid p\in [m]\}$, and $W=\{w_p \setmid p\in [m]\}$.

    \item[(2)] Add edges to make~$V$ a clique.

    \item[(3)] Add a vertex~$z$ adjacent to all vertices in~$T$.
\end{itemize}

\begin{figure}
    \centering
    \includegraphics[width=\textwidth]{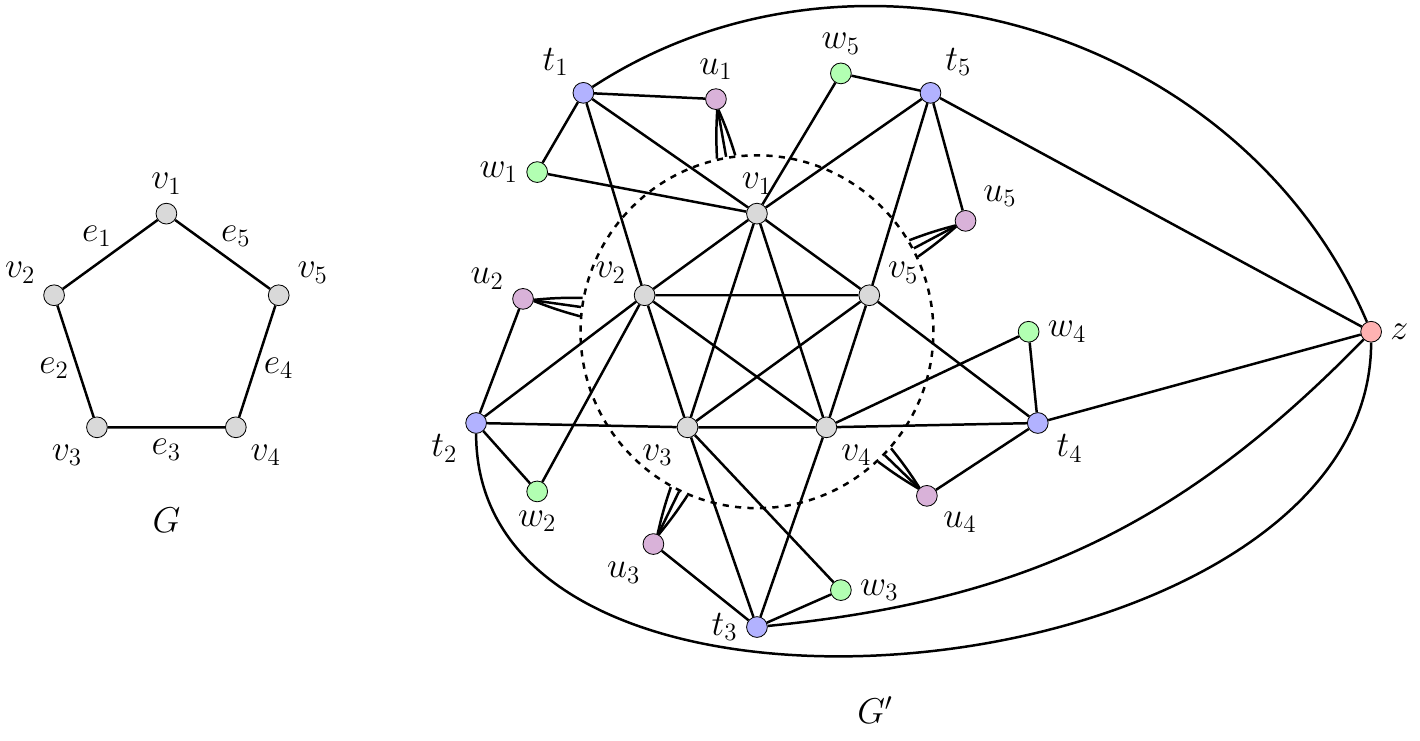}
    \caption{Illustration of the graph~$G'$ constructed in the proof of Theorem~\ref{thm-lexdfs-npc}.
    The three edges between each~$u_x$ and the dashed boundary indicates that~$u_x$ is adjacent to all vertices enclosed by the boundary.
    Vertices belonging to the same set are shown in the same color.
    In the reduction, we assume that~$G$ is $3$-regular and contains at least ten vertices.
    For simplicity, the figure depicts the construction of~$G'$ from a cycle on five vertices.}
    \label{fig-lexdfs-np-hard}
\end{figure}

We now construct a set~$\Pi$ of orderings, each obtained as
 $(\sigma,\textsf{LexDFS}(G';\sigma))$ for a suitable prefix~$\sigma$.
We specify only the prefixes and partition the orderings into the following five groups.
\begin{enumerate}[(i)]
    \item For every pair~$\{i,j\}\subseteq [n]$ with $i<j$, choose an edge~$e_p=\edge{v_x}{v_y}$ of~$G$, where $x<y$, such that~$\{x,y\}\cap\{i,j\}=\emptyset$.
          Such an edge exists since~$G$ admits a matching of size at least three. We add to~$\Pi$ the ordering specified by the prefix
    \[
        \left(v_i, v_j, \overrightarrow{V} \ominus \{v_i, v_j, v_x, v_y\}, v_x, v_y, u_p, t_p, z\right).
    \]

    \item For each edge $e_p = \edge{v_i}{v_j}$ of~$G$, choose an edge $e_q = \edge{v_x}{v_y}$ of~$G$ that shares no endpoint with~$e_p$; such an edge exists by the same argument as before. We add to~$\Pi$ the following three orderings specified by the prefixes:
    \begin{align*}
        \left(t_p, v_i, u_p, v_j, \overrightarrow{V} \ominus \{v_i, v_j, v_x, v_y\}, v_x, v_y, u_q, t_q, z\right),\\
        \left(t_p, v_j, u_p, v_i, \overrightarrow{V} \ominus \{v_i, v_j, v_x, v_y\}, v_x, v_y, u_q, t_q, z\right),\\
        \left(t_p, v_i, w_p, v_j, u_p, \overrightarrow{V} \ominus \{v_i, v_j, v_x, v_y\}, v_x, v_y, u_q, t_q, z\right).
    \end{align*}

    \item For every edge $e_p = \edge{v_i}{v_j}$ in~$G$ where $i < j$, we add to~$\Pi$ an ordering specified by the prefix
    \[
        \left(\overrightarrow{V} \ominus \{v_i, v_j\}, v_i, v_j, u_p, t_p, z\right).
    \]

    \item For every edge $e_p = \edge{v_i}{v_j}$ in~$G$ with $i < j$ and each $v_h \in V \setminus \{v_i, v_j\}$, we add to~$\Pi$ an ordering with the prefix
    \[
        \left(v_i, v_j, \overrightarrow{V} \ominus \{v_i, v_j, v_h\}, v_h, u_p, t_p, z\right).
    \]
\end{enumerate}
In every prefix defined above, the vertex~$z$ appears last.
Let $k=\kappa+\abs{\eset{G'}}$, and define the instance of \prob{Edge-Bounded-LexDFS-SP} as $g(I)=(\vset{G'},\Pi,k)$.
The reduction is computable in polynomial time.
We now prove the correctness of the reduction.

\medskip
\noindent
$(\Rightarrow)$ Assume that~$I$ is a {\yesins} of~\prob{Vertex Cover}, and let~$S$ be a vertex cover of~$G$ of size~$\kappa$. We show that~$g(I)$ is a {\yesins} of~\prob{Edge-Bounded-LexDFS-SP}.

Let~$H$ be obtained from~$G'$ by adding edges between~$z$ and the vertices of~$S$, so that~$H$ has $k = \kappa + \abs{E(G')}$ edges.
We show~$H$ is a LexDFS-support of~$\Pi$. By Lemma~\ref{lem-prefix}, it suffices to show that the specified prefix of each ordering in~$\Pi$ is a partial LexDFS-ordering of~$H$.
We consider the groups of orderings separately.
\begin{itemize}
    \item {Orderings in Group~(i).}

    Consider an ordering in Group~(i) corresponding to a pair~$\{i,j\}\subseteq [n]$ and an edge~$e_p = \{v_x, v_y\}$ with~$\{i,j\}\cap \{x,y\} = \emptyset$ and $x<y$.
    It first visits all vertices in~$V$, starting with~$v_i$ and~$v_j$ and ending with~$v_x$ and~$v_y$, and then visits~$u_p$. Since $V \cup \{u_p\}$ forms a clique in~$H$, the restriction of the ordering to this set is a partial LexDFS-ordering of~$H$.

  The next vertex,~$t_p$, is the only unvisited neighbor of~$u_p$, and can therefore be visited next. At this point,~$t_p$ has two unvisited neighbors,~$w_p$ and~$z$. By construction, the vertex~$w_p$ has exactly two visited neighbors, namely~$t_p$ and~$v_x$. Since~$S$ is a vertex cover,~$z$ is adjacent to at least one of~$v_x$ or~$v_y$ in~$H$.

    \begin{itemize}
        \item If~$z$ is adjacent to~$v_y$ in~$H$, then it has a larger lexicographic label than~$w_p$, so~$z$ should be visited next.
    \item Otherwise,~$z$ is adjacent to~$v_x$. Since~$G$ has a matching of at least three edges, there exists an edge disjoint from~$e_p$ whose endpoints both precede~$v_x$ in the ordering. Thus,~$z$ still has a larger lexicographic label than~$w_p$, and should be visited next.
    \end{itemize}

Hence, each prefix constructed in Group~(i) is a partial LexDFS-ordering of~$H$.

    \item {Orderings in Group~(ii).}


    In the first prefix, the first four vertices~$t_p$,~$v_i$,~$u_p$,~$v_j$ form a clique in~$H$.
    The prefix then visits all remaining vertices of~$V$, ending with $v_x$ and $v_y$, the endpoints of an edge~$e_q$ disjoint from~$\{v_i,v_j\}$.
    Since $V \cup \{u_p\}$ is a clique and no vertex in $V \setminus \{v_i, v_j\}$ is adjacent to the first visited vertex~$t_p$, the prefix remains a partial LexDFS-ordering of~$H$ after all vertices of~$V$ are visited. Then the vertices~$u_q$,~$t_q$, and~$z$ are visited in the order $(u_q,t_q,z)$, where~$u_q$ is adjacent to all vertices of~$V$ in~$H$, whereas~$t_q$ has only two visited neighbors, namely~$v_x$ and~$v_y$. Hence, when it is visited,~$u_q$ has a larger lexicographic label than~$t_q$. Moreover, we claim that at this moment~$u_q$ also has a larger lexicographic label than~$z$. Indeed, if $v_y \notin S$, then~$z$ is not adjacent to~$v_y$. Otherwise, since $\kappa \leq n-1$, there exists a vertex in~$V$ not adjacent to~$z$, and additionally~$z$ is not adjacent to~$u_p$. In both cases,~$u_q$ has a larger label than~$z$.
Therefore, the first prefix is a partial LexDFS-ordering of~$H$.

   Now, consider the second prefix, which is obtained from the first by swapping~$v_i$ and~$v_j$.
The same analysis applies, and therefore it is also a partial LexDFS-ordering of~$H$.

    We now turn to the third prefix. The first three visited vertices~$t_p$,~$v_i$, and~$w_p$ form a triangle in~$H$. The vertex~$w_p$ has no unvisited neighbors. The next vertex~$v_j$ is adjacent to~$t_p$ and~$v_i$ in~$H$, and hence must have the maximum lexicographic label at that point. Then,~$u_p$, being adjacent to~$t_p$,~$v_i$, and~$v_j$, has the maximum label and is visited next. The remainder of the analysis is analogous to the discussion for the first prefix.

    \item {Orderings in Group~(iii).}

    Each ordering in this group, corresponding to an edge~$e_p$, first visits the vertices in $V \cup \{u_p\}$, which form a clique in~$H$, with~$u_p$ visited last. Then~$t_p$, the only unvisited neighbor of~$u_p$, is visited, followed by~$z$. Thus, the ordering is a partial LexDFS-ordering of~$H$.

    \item {Orderings in Group~(iv).}

    These orderings are similar to those in Group~(iii), differing only in the relative positions of vertices in~$V$. The same argument applies, and each such ordering is a partial LexDFS-ordering of~$H$.
\end{itemize}

This completes the proof that~$g(I)$ is a {\yesins} of \prob{Edge-Bounded-LexDFS-SP}.

$(\Leftarrow)$ Assume that~$g(I)$ is a {\yesins}. Let~$H$ be a LexDFS-support of~$\Pi$ with the minimum number of edges. We have $\abs{\eset{H}} \leq k = \kappa + \abs{\eset{G'}}$.
To prove that $G$ admits a vertex cover of size $\kappa$, we first prove the following lemma.

\begin{claim}
\label{claim-lexdfs-a}
$\eset{G'} \subseteq \eset{H}$.
\end{claim}

\begin{proof}
We prove the claim by showing that all edges of~$G'$ constructed in Groups~(1)--(3) are contained in~$H$.
We first establish the result for the edges in Group~(2), as this will be used in the proofs for the remaining groups.

\begin{itemize}

  \item {Edges constructed in~(2), i.e., edges among vertices in~$V$.}

    By the construction of orderings in Group~(i), for every pair $\{i,j\} \subseteq [n]$, there exists an ordering in~$\Pi$ that starts with the two vertices~$v_i$ and~$v_j$. Since this ordering is a LexDFS-ordering of~$H$, Observation~\ref{ob-connected} implies that~$v_i$ and~$v_j$ are adjacent in~$H$. Therefore, all edges among the vertices in~$V$ exist in~$H$.

   \item {Edges constructed in~(1), i.e., those introduced for edges in~$G$.}

    Let $e_p=\edge{v_i}{v_j}$ be an edge in $G$ with $i<j$. By the construction of the orderings in Group~(ii), there exists one ordering in which $t_p$ and~$v_i$ are the first two vertices, and another in which $t_p$ and~$v_j$ are the first two vertices. This implies that the two edges between~$\{t_p\}$ and~$\{v_i, v_j\}$ constructed in~(1) are contained in~$H$.

Now consider the first ordering constructed in~(ii), focusing on the three vertices~$v_i$,~$u_p$, and~$v_j$, which appear in the second, third, and fourth positions, respectively. We show that $\edge{v_i}{u_p} \in \eset{H}$.
Suppose for  contradiction that~$v_i$ and~$u_p$ are not adjacent in~$H$. As proved above,~$v_i$ and $v_j$ are adjacent in $H$. Then, by the four-point ordering characterization of LexDFS (Lemma~\ref{lem-four-vertex}, Statement~\ref{four-point-lexdfs}), there exists a vertex between~$v_i$ and~$u_p$ in the ordering that is adjacent to~$u_p$ but not adjacent to~$v_j$ in~$H$. However, no such vertex exists between~$v_i$ and~$u_p$ in the ordering, yielding a contradiction.

By an analogous argument applied to the second ordering constructed in~(ii), we can conclude that~$v_j$ and~$u_p$ are also adjacent in~$H$.

More generally, the above reasoning implies the following fact.

    \medskip
    \noindent\textbf{Fact 1:} {\textit{Let~$\sigma$ be a LexDFS-ordering of~$H$, and let $a <_\sigma b <_\sigma c$ be vertices such that no vertex lies between~$a$ and~$b$ in~$\sigma$. If~$\edge{a}{c}\in \eset{H}$, then~$\edge{a}{b}\in \eset{H}$.}}
    \medskip

    We now show that~$\edge{t_p}{u_p}\in\eset{H}$. Suppose, for contradiction, that~$t_p$ and~$u_p$ are not adjacent in~$H$. Consider the prefix~$(t_p,v_i,u_p,v_j)$ of the first ordering in Group~(ii). Recall that~$t_p$ is adjacent to~$v_j$ in~$H$. By Lemma~\ref{lem-four-vertex}, the only vertex between~$t_p$ and~$u_p$, namely~$v_i$, is adjacent to~$u_p$ but not to~$v_j$. However, we have already shown that~$v_i$ and~$v_j$ are adjacent in~$H$, a contradiction.

Next, we show that $\edge{w_p}{v_i}, \edge{w_p}{t_p} \in \eset{H}$. Consider the third ordering in~(ii), and focus on the three vertices~$v_i$,~$w_p$, and~$v_j$, which appear in the second, third, and fourth positions, respectively. Since $\edge{v_i}{v_j} \in \eset{H}$, by Fact~1, we obtain $\edge{w_p}{v_i} \in \eset{H}$. Now consider the three vertices~$t_p$,~$w_p$, and~$v_j$ in the same ordering, where~$t_p$ appears in the first position. Suppose, for contradiction, that~$w_p$ and~$t_p$ are not adjacent in~$H$. Since $\edge{t_p}{v_j} \in \eset{H}$, by Lemma~\ref{lem-four-vertex}, the only vertex between~$t_p$ and~$w_p$, namely~$v_i$, must be adjacent to~$w_p$ and nonadjacent to~$v_j$. However, we have already established that $\edge{v_i}{v_j} \in \eset{H}$, yielding a contradiction. Hence, $\edge{w_p}{t_p} \in \eset{H}$.

We now show that~$u_p$ is adjacent to all vertices in~$V$. We have already established that~$u_p$ is adjacent to both~$v_i$ and~$v_j$ in~$H$. It remains to prove that $\edge{u_p}{v_h} \in \eset{H}$ for every $h \in [n] \setminus \{i, j\}$. Let $h \in [n] \setminus \{i, j\}$, and consider the ordering as defined in Group~(iv). Since~$G$ is $3$-regular, there exists an edge~$e_q$ incident to~$v_h$ in~$G$. By the definition of this ordering,~$u_p$ appears between $v_h$ and the vertex~$w_q$, one of the vertices created for the edge~$e_q$. Moreover, no vertices are between~$v_h$ and~$u_p$ in the ordering. Note that, by the previous analysis, $w_q$ is adjacent to $v_h$ in~$H$.
Applying Fact~1 to the three vertices~$v_h$,~$u_p$, and~$w_q$, we conclude that~$u_p$ is adjacent to~$v_h$ in~$H$. This completes the proof that~$u_p$ is adjacent to all vertices in~$V$.

\item {Edges constructed in~(3), i.e., those incident to~$z$.}

Finally, we show that~$z$ is adjacent to all vertices in $T$. Consider the ordering defined in Group~(iii), constructed for an edge~$e_p$. Focus on the three vertices~$t_p$,~$z$, and~$w_p$. By construction, $z$ appears between the other two vertices, and no vertex lies between $t_p$ and $z$ in the ordering. From the previous analysis, we know that~$w_p$ is adjacent to~$t_p$ in~$H$. Then, by Fact~1,~$z$ must also be adjacent to~$t_p$ in~$H$, completing the proof.
\end{itemize}
Putting all the above together, we conclude that $\eset{G'} \subseteq \eset{H}$.
\end{proof}

By Claim~\ref{claim-lexdfs-a},~$z$ has at most
$\abs{\eset{H}} - \abs{\eset{G'}} \leq \kappa$ neighbors in
$\vset{G'} \setminus T = V \cup W \cup U$, i.e.,
\begin{equation}
    \label{eq-lexdfs}
    \abs{\neighbor{z}{H} \cap (V \cup W \cup U)} \leq \kappa.
\end{equation}

We now construct a vertex cover of~$G$ as follows. Let
$S = \neighbor{z}{H} \cap V$.
If~$S$ is a vertex cover of~$G$, we are done.
Otherwise, consider an edge
$e_p = \edge{v_i}{v_j}$ of~$G$ with $i < j$ not covered by~$S$.
Consider the ordering constructed for~$e_p$ in Group~(iii).
By Claim~\ref{claim-lexdfs-a},~$v_i$ and~$w_p$ are adjacent in~$H$.
In the corresponding ordering,~$v_i$ precedes~$w_p$, and~$z$ appears between them.
Since $v_i \notin S$, it is not adjacent to~$z$ in~$H$.
By Lemma~\ref{lem-four-vertex}, there must exist a vertex between~$v_i$ and~$z$ in the ordering
that is adjacent to~$z$ but not adjacent to~$w_p$ in~$H$.
This vertex cannot be~$v_j$, as~$v_j\not\in S$.
It also cannot be~$t_p$, since~$t_p$ and~$w_p$ are adjacent in~$G'$ and remain adjacent in~$H$ by Claim~\ref{claim-lexdfs-a}.
The only remaining vertex between~$v_i$ and~$z$ is~$u_p$, so we conclude that~$z$ must be adjacent to~$u_p$ in~$H$.

Therefore, if an edge $e_p$ of~$G$ is not covered by~$S$, then the corresponding vertex~$u_p \in U$, introduced for this edge, is adjacent to~$z$ in~$H$.
It follows that the number of edges of~$G$ not covered by~$S$ is bounded from above by $\abs{\neighbor{z}{H} \cap U}$.
Based on this observation, for every edge not covered by~$S$ in~$G$, we expand~$S$ by adding an arbitrary endpoint of the edge.
After this process,~$S$ becomes a vertex cover of~$G$.
Furthermore, the size of~$S$ is bounded by
\[
\abs{\neighbor{z}{H} \cap V} + \abs{\neighbor{z}{H} \cap U}
\leq \abs{\neighbor{z}{H} \cap (V \cup W \cup U)},
\]
which, by Inequality~\eqref{eq-lexdfs}, is at most~$\kappa$.
We conclude that~$I$ is a {\yesins} of \prob{Vertex Cover}.
\end{proof}

By a slight adaptation of the above reduction, we obtain the following result.

\begin{theorem}
\label{thm-lexdfs-npc-delta}
\prob{Deg-Bounded-LexDFS-SP} is {\nph}.
\end{theorem}

\begin{proof}
We employ the same reduction as in the proof of Theorem~\ref{thm-lexdfs-npc}, with the only modification that we set~$k$, the maximum degree bound, to~$m + \kappa$. The correctness of the reduction remains essentially unchanged.
For the $(\Rightarrow)$ direction, assume that the instance~$I$ of \prob{Vertex Cover} is a {\yesins}, that is, there exists a vertex cover $S \subseteq V$ of size~$\kappa$ in~$G$. As shown previously, the graph~$H$ obtained from~$G'$ by adding the~$\kappa$ edges between~$z$ and the vertices in~$S$ is a LexDFS-support of~$\Pi$. In this graph~$H$, the vertex~$z$ is adjacent to all vertices in~$S \cup T$, and hence has degree~$\kappa + m = k$.
For the $(\Leftarrow)$ direction, the argument is identical to that in the proof of Theorem~\ref{thm-lexbfs-npc}.
\end{proof}

Now, we study the graph search MCS.

\begin{theorem}
\label{thm-mcs-edges-npc}
	\prob{\prob{Edge-Bounded-MCS-SP}} is {\nph}.
\end{theorem}

\begin{proof}
Let $I = (G, \kappa)$ be an instance of  \prob{Vertex Cover} on $3$-regular graphs. 
\begin{figure}
    \centering
    \includegraphics[width=0.5\textwidth]{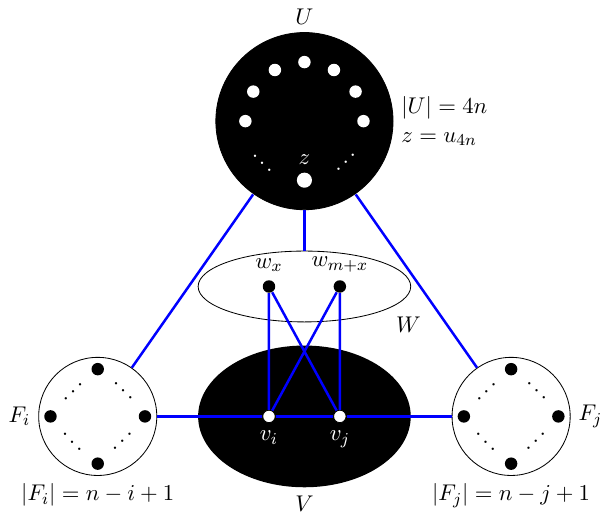}
    \caption{An illustration of the graph~$G'$ in the proof of Theorem~\ref{thm-mcs-edges-npc}. Vertices in the shaded regions represent cliques, while vertices in the white regions form independent sets.}
    \label{fig:enter-label}
\end{figure}
Let~$G'$ be the graph obtained from~$G$ by performing the following operations:
\begin{enumerate}
    \item[(1)] For every edge $e_x = \edge{v_i}{v_j}$, $x\in [m]$, $\{i,j\}\subseteq [n]$, introduce two vertices~$w_x$ and~$w_{m+x}$, and add all four edges between $\{v_i, v_j\}$ and $\{w_x, w_{m+x}\}$. Let $W = \{w_i \setmid i \in [2m]\}$ and let $\overrightarrow{W} = (w_1, w_2, \ldots, w_{2m})$. For each $i \in [n]$ define $W(i) = \neighbor{v_i}{G'} \cap W$, and for every pair $\{i, j\} \subseteq [n]$ define $W[i,j] = W(i) \cup W(j)$.

    \item[(2)] Add edges to make~$V$ a clique.

    \item[(3)] For every $i \in [n]$, introduce a set $F_i = \{f_i^1, f_i^2, \ldots, f_i^{n-i+1}\}$ of $n-i+1$ vertices adjacent only to~$v_i$. Let $\overrightarrow{F_i} = (f_i^1, f_i^2, \ldots, f_i^{n-i+1})$. Let $F = \bigcup_{i \in [n]} F_i$ and $\overrightarrow{F} = (\overrightarrow{F_1}, \overrightarrow{F_2}, \ldots, \overrightarrow{F_n})$.

    \item[(4)] Add a clique $U=\{u_x \setmid x\in[4n]\}$ of $4n$ vertices, and make every vertex of~$U$ adjacent to every vertex in~$W\cup F$. Let $\overrightarrow{U}=(u_1,u_2,\ldots,u_{4n})$, and in particular let~$z=u_{4n}$.
\end{enumerate}
Notice that $W \cup F$ is an independent set in~$G'$. Figure~\ref{fig:enter-label} illustrates the construction of~$G'$.

We now construct a set~$\Pi$ of orderings. They are designed so that the first two vertices and the connectivity properties of the graph traversal ensure that any graph admitting these orderings as MCS-orderings must contain~$G'$ as a subgraph. Moreover, they remain valid MCS-orderings for any graph obtained from~$G'$ by adding edges between~$z$ and any vertex cover of~$G'$.
Specifically,~$\Pi$ consists of the following six groups of orderings:
\begin{enumerate}
    \item[(i)] For each $\{i, j\} \subseteq [n]$ with $i < j$, add:
    \[
    \left(v_i, v_j, \overrightarrow{V} \ominus \{v_i,v_j\}, z, w_1, w_{m+1}, \overrightarrow{U} \ominus \{z\}, \overrightarrow{W} \ominus \{w_1, w_{m+1}\}, \overrightarrow{F}\right).
    \]

    \item[(ii)] For each $\{i, j\} \subseteq [4n]$ with $i < j$, add:
    \[
    \left(u_i, u_j, \overrightarrow{U} \ominus \{u_i, u_j\}, \overrightarrow{W}, \overrightarrow{F}, \overrightarrow{V}\right).
    \]

    \item[(iii)] For each $i \in [2m]$ and $j \in [4n]$, add:
    \[
    \left(w_i, u_j, \overrightarrow{U} \ominus \{u_j\}, \overrightarrow{W} \ominus \{w_i\}, \overrightarrow{F}, \overrightarrow{V}\right),
    \]

    \item[(iv)] For every $i \in [n]$, $p \in [n-i+1]$, and $x \in [4n]$, add:
    \[
    \left(f_i^p, u_x, \overrightarrow{U} \ominus \{u_x\}, \overrightarrow{F} \ominus \{f_i^p\}, \overrightarrow{W}, \overrightarrow{V}\right).
    \]

    \item[(v)] For each $i \in [n]$, $p \in [n-i+1]$, add:
    \[
    \left(v_i, f_i^p, z, \overrightarrow{U}\ominus \{z\}, \overrightarrow{W(i)},  \overrightarrow{F_i}\ominus \{f_i^p\}, \overrightarrow{W} \ominus W(i), \overrightarrow{F} \ominus F_i, \overrightarrow{V} \ominus \{v_i\}\right).
    \]

    \item[(vi)] For each edge $e_x = \edge{v_i}{v_j}$, where $x\in [m]$, $\{i,j\}\subseteq [n]$, and $i < j$, add the following four orderings, each followed by the common suffix. The prefixes are
    \[
    \begin{aligned}
    &(w_x, v_i, v_j, z, w_{m+x}), \\
    &(w_{m+x}, v_i, v_j, z, w_x), \\
    &(w_x, v_j, v_i, z, w_{m+x}), \\
    &(w_{m+x}, v_j, v_i, z, w_x).
    \end{aligned}
    \]
    The common suffix is
    \[
    \left(\overrightarrow{U} \ominus \{z\}, \overrightarrow{W[i,j]} \ominus \{w_x, w_{m+x}\}, \overrightarrow{W} \ominus W[i,j], \overrightarrow{F_i}, \overrightarrow{F_j}, \overrightarrow{F} \ominus (F_i\cup F_j), \overrightarrow{V} \ominus \{v_i, v_j\}\right).
    \]
\end{enumerate}
	%
%
Let $k=\kappa+\abs{\eset{G'}}$, and and define the instance $g(I)=(\vset{G'}, \Pi, k)$ of {\prob{Edge-Bounded-MCS-SP}}.
We prove the correctness of the reduction as follows.

$(\Rightarrow)$ Assume that~$G$ has a vertex cover~$S$ of~$\kappa$ vertices.
Let~$H$ be obtained from~$G'$ by adding edges between~$z$ and the vertices of~$S$. Then,~$H$ has exactly $k$ edges.
We show that~$H$ is an MCS-support of~$\Pi$, and hence that~$g(I)$ is a {\yesins}. We consider the six groups of orderings separately.

\begin{itemize}
    \item {Orderings in group (i).}

    Let~$\sigma$ denote the ordering constructed for a pair $\{i,j\}\subseteq[n]$ with $i<j$.
    Since the first~$n$ vertices of~$\sigma$ form the clique~$V$,~$\sigma_{\le n}$ is a partial MCS-ordering of~$H$.
    After the vertices of~$V$ have been visited,~$z$ has~$\kappa$ visited neighbors (those in~$S$), each vertex in~$W$ has two visited neighbors, each vertex in~$F$ has one visited neighbor, and no vertex in $U\setminus \{z\}$ has a visited neighbor.
    Since~$\kappa \geq 3$, the vertex~$z$ has more visited neighbors than any other unvisited vertex and is therefore selected next. After visiting~$z$, each vertex in~$W$ has three visited neighbors, whereas every unvisited vertex outside~$W$ has at most two visited neighbors. The ordering selects~$w_1$ and~$w_{m+1}$ as the next to be visited (recall that~$W$ is an independent set in~$H$). After visiting~$w_1$ and~$w_{m+1}$, each vertex in~$U\setminus \{z\}$ has three visited neighbors~$w_1$,~$w_{m+1}$ and $z$.
    The ordering may now visit the vertices of~$U$. Since~$U$ is a clique, once the first vertex of~$U$ is selected, all remaining vertices of~$U$ can be visited consecutively while preserving the MCS property.
     After all vertices of~$U$ have been visited, each vertex in~$W\setminus \{w_1, w_{m+1}\}$ has $\abs{U} + 2$ visited neighbors, whereas each vertex in~$F$ has $\abs{U}+1$ visited neighbors. Since $W \cup F$ is an independent set,~$\sigma$ can visit the remaining vertices in~$W$ first, followed by all vertices in~$F$, without ceasing to be a partial MCS-ordering of~$H$ at any point.

    \item {Orderings in group (ii).}

   Consider the ordering constructed for a pair $\{i,j\}\subseteq [4n]$ with~$i<j$. It begins by visiting all vertices of the clique~$U$, starting with~$u_i$ and~$u_j$.
   It then visits the vertices of the independent set~$W \cup F$, each of which is adjacent to all vertices in~$U$. Consequently, when any vertex from~$W \cup F$ is visited, it already has $\abs{U}=4n$ visited neighbors.
Each vertex~$v_x\in V$, $x\in [n]$, has at most one neighbor in~$U$ (namely~$z$, if~$v_x\in S$), at most six neighbors in~$W$ (since $G$ is $3$-regular), and at most~$n-x+1$ neighbors in~$F$. Hence it has at most~$n-x+8\le n+7$ visited neighbors. Since~$n\ge4$, every vertex in~$W\cup F$ has more visited neighbors than any vertex in~$V$, and thus the subordering restricted to $U \cup W \cup F$ constitutes a partial MCS-ordering of~$H$.
%

After all vertices in~$U\cup W\cup F$ have been visited, the vertices of~$V$ are visited as $(v_1,\ldots,v_n)$. Each $v_x$ where~$x\in[n]$ has either~$n-x+7$ neighbors in~$U\cup W\cup F$ if~$v_x\notin S$, or~$n-x+8$ such neighbors if~$v_x\in S$. Together with the~$x-1$ previously visited vertices of~$V$, it follows that~$v_x$ has at least~$n+6$ visited neighbors when selected. On the other hand, every unvisited vertex~$v_y$ with~$y>x$ has at most~$(n-y+8)+(x-1)\leq n+6$ visited neighbors. Hence the ordering remains an MCS-ordering throughout the visit of the vertices in~$V$.
%

\item {Orderings in Group~(iii).}

Consider the ordering constructed for some $i\in[2m]$ and $j\in[4n]$. It starts with~$w_i$ and~$u_j$, then visits the remaining vertices of~$U$, followed by the remaining vertices of~$W$.
Since every vertex of~$U$ is adjacent to every vertex of~$W$,~$U$ is a clique, and~$W$ is an independent set, the ordering may visit all vertices of~$U$ followed by vertices of~$W$ while preserving the MCS property.
The ordering then visits the vertices of~$F$, followed by the vertices of~$V$ in the order $(v_1,\ldots,v_n)$. At this point, the situation is identical to that in the proof for Group~(ii). Hence the same argument shows that the remainder of the ordering is consistent with MCS.

   \item {Orderings in group~(iv).}

  Consider the ordering constructed for some $i \in [n]$, $p \in [n - i + 1]$, and $x \in [4n]$.
It starts with~$f_i^p$ and~$u_x$, and then visits the remaining vertices of $U$.
Since~$U\cup\{f_i^p\}$ is a clique in~$H$, the ordering remains a partial MCS-ordering throughout the visit of these vertices.
It then visits the vertices in the independent set $(F \setminus \{f_i^p\}) \cup W$.
Every vertex in this set is adjacent to all vertices of~$U$. Hence the ordering remains a partial MCS-ordering throughout this stage.
Finally, vertices of~$V$ are visited in  $(v_1, v_2, \dots, v_n)$.
The same argument as in the proof for Group~(ii) shows that the remainder of the ordering is consistent with MCS.

\item {Orderings in group~(v).}

Consider the ordering~$\sigma$ constructed for some $i \in [n]$ and $p \in [n-i+1]$. It first visits~$v_i$, followed by its neighbor~$f_i^p$, and then all vertices of~$U$, starting with~$z=u_{4n}$.
When~$z$ is visited, it is adjacent to~$f_i^p$ and possibly also to~$v_i$, whereas no other unvisited vertex is adjacent to both~$v_i$ and~$f_i^p$. Therefore,~$\sigma_{\le 3}$ is a partial MCS-ordering of~$H$. Moreover, since~$U\cup\{f_i^p\}$ is a clique in~$H$, the ordering remains a partial MCS-ordering throughout the visit of the vertices in~$U\setminus\{z\}$.

Next,~$\sigma$ visits the vertices in the independent set $W(i)\cup F_i \setminus \{f_i^p\}$.
Each such vertex has $\{v_i\} + \abs{U} = 4n + 1$ visited neighbors, and no remaining vertex has more.
Therefore, the ordering remains a partial MCS-ordering of~$H$ throughout the visit of these vertices.

Afterward,~$\sigma$ visits the vertices in the independent set $(W\setminus W(i))\cup (F\setminus F_i)$. Each such vertex has~$\abs{U}=4n$ visited neighbors, which is maximum among the remaining vertices. Thus, the ordering remains a partial MCS-ordering throughout the visit of these vertices.

Finally, the vertices in~$V \setminus \{v_i\}$ are visited in increasing order of their indices.
Consider the moment when some vertex~$v_j$, where $j \in [n] \setminus \{i\}$, is about to be visited.
Since~$V$ is a clique in~$H$, $v_j$ has $j - 1$ visited neighbors in~$V$ if $j > i$, and $j$ visited neighbors otherwise.
Moreover, it has $\abs{F_j} = n - j + 1$ visited neighbors in~$F$, and six visited neighbors in~$W$.
It also has one visited neighbor in~$U$ if $v_j \in S$, and none otherwise.
The number of visited neighbors of any vertex~$v_{j'}$ with $j'>j$ at the moment when~$v_j$ is visited can be similarly calculated. Table~\ref{tab-a} summarizes their visited-neighbor counts.
\begin{table}[ht!]
\begin{center}
\begin{tabular}{|c|c|c|}
\hline
 & $j < i$ & $j > i$ \\ \hline
$v_j \in S$ & $n + 8$ & $n + 7$ \\ \hline
$v_j \notin S$ & $n + 7$ & $n + 6$ \\ \hline
\end{tabular}\ \ \ \ \
\begin{tabular}{|c|c|c|}
\hline
 & $j' < i$ & $j' > i$ \\ \hline
$v_{j'} \in S$ & $n + 8 + j - j'$ & $n + 7 + j - j'$ \\ \hline
$v_{j'} \notin S$ & $n + 7 + j - j'$ & $n + 6 + j - j'$ \\ \hline
\end{tabular}
\end{center}
\caption{Visited-neighbor counts of~$v_j$ (left) and~$v_{j'}$ with~$j'>j$ (right) when~$v_j$ is about to be visited.}
\label{tab-a}
\end{table}
A straightforward verification shows that every vertex~$v_{j'}$ with $j'>j$ has at most as many visited neighbors as~$v_j$. Therefore,~$v_j$ is a valid MCS choice when selected, and the ordering remains a partial MCS-ordering of~$H$ throughout the visit of the vertices in~$V\setminus\{v_i\}$.
%

\item {Orderings in group (vi).}

This group contains, for every edge $e_x = \edge{v_i}{v_j}$ with $i < j$, four orderings. By symmetry, it suffices to show that the first ordering is an MCS-ordering of~$H$.

Let~$\sigma$ denote the first ordering, whose prefix is $(w_x, v_i, v_j, z, w_{m+x})$. Since~$w_x$,~$v_i$, and~$v_j$ form a clique in~$H$,~$\sigma_{\le 3}$ is a partial MCS-ordering of~$H$.

Since~$S$ is a vertex cover of~$G$,~$z$ is adjacent to at least one of~$v_i$ and~$v_j$. Together with its adjacency to~$w_x$, this implies that~$z$ has at least two visited neighbors when selected. No unvisited vertex is adjacent to all three of~$w_x$,~$v_i$, and~$v_j$, and therefore~$\sigma_{\le 4}$ is a partial MCS-ordering of~$H$.

The vertex~$w_{m+x}$ has three visited neighbors $v_i$,~$v_j$, and~$z$ when selected, and no unvisited vertex has more. Hence~$\sigma_{\le 5}$ is a partial MCS-ordering of~$H$.

A routine verification analogous to the previous cases shows that the remaining steps in~$\sigma$ preserve the MCS-ordering property.  
\end{itemize}

We have shown that~$H$ has~$k$ edges and admits every ordering in~$\Pi$ as an MCS-ordering. Therefore,~$g(I)$ is a {\yesins}.

$(\Leftarrow)$ Assume that~$g(I)$ is a {\yesins}, and let~$H$ be an MCS-support of~$\Pi$ with the minimum number of edges.
Then $\abs{\eset{H}} \leq k$.
For every edge of~$G'$, we constructed an ordering in~$\Pi$ whose first two vertices are the endpoints of that edge.
It follows that $\eset{G'} \subseteq \eset{H}$.

Now consider any edge $e_x = \edge{v_i}{v_j}$ in~$G$.
In the first ordering of Group~(vi) associated with~$e_x$, the first five vertices are $(w_x, v_i, v_j, z, w_{m+x})$.
After~$v_j$ is visited, the vertex~$w_{m+x}$ has exactly two visited neighbors, namely~$v_i$ and~$v_j$.
Since~$z$ is visited before~$w_{m+x}$ and the ordering is an MCS-ordering of~$H$, the vertex~$z$ must also have at least two visited neighbors at that moment.
It follows that~$z$ is adjacent to at least one of~$v_i$ and~$v_j$.
Let $S=\neighbor{z}{H} \cap V$.
Since~$e_x$ was arbitrary, $S$  is a vertex cover of~$G$.
Moreover, as~$z$ has no neighbor in~$V$ in~$G'$, every vertex in~$S$ contributes a distinct edge in~$\eset{H}\setminus\eset{G'}$.  Hence $\abs{S} \leq k - \abs{\eset{G'}} = \kappa$.
This establishes that~$I$ is a {\yesins} of {\prob{Vertex Cover}}.
\end{proof}

By appropriately adjusting the value of $k$ in the above reduction, we obtain the following result.

\begin{theorem}
\label{thm-mcs-degree-npc}
	\prob{Deg-Bounded-MCS-SP} is {\nph}.
\end{theorem}

\begin{proof}
Set $k = 2m + 4n - 1+\frac{n\cdot (n+1)}{2}+ \kappa$
which equals the degree of~$z$ in~$G'$ plus~$\kappa$. Since every additional edge used in the proof of Theorem~\ref{thm-mcs-edges-npc} is incident with~$z$, the same reduction and correctness argument apply verbatim.
\end{proof}

We now turn our attention to the final graph search, MNS, considered in this paper.

\begin{theorem}
\label{thm-mns-edges-npc}
	\prob{Edge-Bounded-MNS-SP} is {\nph}.
\end{theorem}

\begin{proof}
Let $I=(G, \kappa)$ be an instance of {\prob{Vertex Cover}} on $3$-regular graphs. We assume that~$G$ is connected.
Define~$G'$ as the graph obtained from~$G$ by performing the following operations:
\begin{enumerate}[(1)]
\item For each edge of~$G$, introduce three vertices and partition them into two sets. Precisely, we introduce a set $T=\{t_x \setmid x\in [m]\}$ of~$m$ vertices, and introduce a set $W =\{w_i \setmid i\in [2m]\}$ of~$2m$ vertices.
	\item For every edge $e_x=\edge{v_i}{v_j}$ in~$G$, where $x\in [m]$, $\{i,j\}\subseteq [n]$, add the six edges $\edge{v_i}{w_x}$, $\edge{w_x}{t_x}$,  $\edge{t_x}{w_{m+x}}$,
        $\edge{w_{m+x}}{v_j}$,
         $\edge{v_j}{t_x}$, and $\edge{t_x}{v_i}$. See Figure~\ref{fig-mns-a}.
       \begin{figure}[H]
       \centering
        \includegraphics[width=0.25\textwidth]{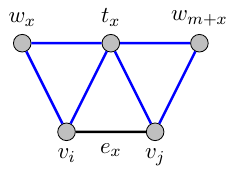}
        \caption{An illustration of the addition of the six edges (colored in blue) in (2).}
        \label{fig-mns-a}
    \end{figure}

	\item Add edges to make~$V$ a clique.

	\item For every $i \in [n]$, add a vertex~$f_i$, and add an edge between~$f_i$ and~$v_i$.
    Let $F = \{f_i \setmid i \in [n]\}$.

	\item Add a clique $U=\{u_i \setmid i\in[3n]\}$ of~$3n$ vertices. Let $z = u_{3n}$. Make every vertex of~$U$ adjacent to every vertex of~$T$, and every vertex of~$U\setminus\{z\}$ adjacent to every vertex of~$F$.
\end{enumerate}

Observe that~$W$ is an independent set in~$G'$ and every vertex of~$W$ has degree two.

Let $\overrightarrow{W} = (w_1, w_2, \ldots, w_{2m})$, $\overrightarrow{T} =(t_1, t_2, \ldots, t_m)$, $\overrightarrow{U} = (u_1, u_2, \ldots, u_{3n})$, and $\overrightarrow{F} = (f_1, f_2, \ldots, f_n)$.  For each $i \in [n]$, let $T(i) = \{t_x \setmid x\in [m], v_i\in e_x\}$ denote the set of vertices in~$T$ corresponding to edges of~$G$ incident to~$v_i$. For  $\{i,j\}\subseteq [n]$, let $T(i,j)=T(i)\cup T(j)$, and let $\overrightarrow{T(i,j)}$ be an arbitrary  ordering of~$T(i,j)$.
Define $k=\kappa+\abs{\eset{G'}}$. We construct a set~$\Pi$ of orderings partitioned into six groups:
\begin{enumerate}[(i)]
	\item For each pair $\{i, j\} \subseteq [n]$ with $i < j$, add the ordering
	\[\left(v_i,v_j, \overrightarrow{V} \ominus \{v_i,v_j\}, z, \overrightarrow{T}, \overrightarrow{U}\ominus\{z\}, \overrightarrow{F}, \overrightarrow{W}\right).\]
	
	\item For each pair $\{i, j\} \subseteq [3n]$ with $i < j$, add the ordering
	\[\left(u_i, u_j, \overrightarrow{U} \ominus \{u_i, u_j\}, \overrightarrow{T}, \overrightarrow{F}, \overrightarrow{V}, \overrightarrow{W}\right).\]
	
	\item For each $i \in [m]$ and each $j \in [3n]$, add the ordering
	\[\left(t_i, u_j, \overrightarrow{U} \ominus \{u_j\}, \overrightarrow{T} \ominus \{t_i\}, \overrightarrow{F}, \overrightarrow{V}, \overrightarrow{W}\right).\]
	
	\item For each $i \in [n]$ and each $j \in [3n-1]$, add  the ordering
	\[\left(f_i, u_j, \overrightarrow{U} \ominus \{u_j, z\}, z, \overrightarrow{T}, \overrightarrow{F}\ominus\{f_i\}, \overrightarrow{V}, \overrightarrow{W}\right).\]
	
	\item For each $i \in [n]$, add the ordering
	\[\left(v_i, f_i, \overrightarrow{U} \ominus \{z\}, \overrightarrow{T(i)}, z, \overrightarrow{T} \ominus T(i), \overrightarrow{F} \ominus \{f_i\}, \overrightarrow{V} \ominus\{v_i\}, \overrightarrow{W}\right).\]
	
	\item For each edge $e_x = \edge{v_i}{v_j}$ in~$G$, where $x \in [m]$ and $\{i, j\} \subseteq [n]$ with $i < j$, we define \[\Omega=\left(\overrightarrow{U} \ominus\{z\}, \overrightarrow{T(i,j)} \ominus \{t_x\},
    \overrightarrow{T}\ominus {T(i,j)}, f_i, f_j, \overrightarrow{F}\ominus\{f_i, f_j\},
   \overrightarrow{V} \ominus \{v_i, v_j\}\right),\]
and add the following six orderings:
\begin{itemize}
    \item $\left(w_x, t_x, v_i, v_j, z, \Omega, \overrightarrow{W} \ominus \{w_x\}\right)$
	
        \item $\left(w_x, v_i, t_x, v_j, z, \Omega, \overrightarrow{W}\ominus \{w_x\}\right)$

     \item $\left(w_{m+x}, t_x, v_j, v_i, z,  \Omega, \overrightarrow{W}\ominus \{w_{m+x}\}\right)$

	    \item $\left(w_{m+x}, v_j, t_x, v_i, z,\Omega, \overrightarrow{W}\ominus \{w_{m+x}\}\right)$
	
    \item $\left(t_x, v_i, v_j, z, \Omega, \overrightarrow{W}\right)$
	
	\item $\left(t_x, v_j, v_i, z, \Omega, \overrightarrow{W}\right)$.
\end{itemize}
	%
   %
 %
	%
	%

The purpose of this group is to force every MNS-support of~$\Pi$ to contain all edges in~$G'$ introduced for the edges in~$G$ (the blue edges in Figure~\ref{fig-mns-a}). Only the displayed prefixes are relevant; the remaining vertices are ordered according to an  MNS-ordering of~$G'$ extending the specified prefix.
\end{enumerate}
Let $g(I)=(\vset{G'}, \Pi, k)$.
We show that~$I$ is a {\yesins} of {\prob{Vertex Cover}} if and only if~$g(I)$ is  a {\yesins} of \prob{Edge-Bounded-MNS-SP}.

$(\Rightarrow)$ Assume that~$G$ has a vertex cover~$S$ of~$\kappa$ vertices. Let~$H$ be the graph obtained from~$G'$ by adding the~$\kappa$ edges between~$z$ and the vertices in~$S$. To prove that~$g(I)$ is a {\yesins}, it suffices to show that every ordering in~$\Pi$ is an MNS-ordering of~$H$.
We analyze the six groups of orderings separately.

\begin{itemize}
    \item Orderings in group~(i).

    Consider the ordering~$\sigma$ constructed for a pair $\{i, j\} \subseteq [n]$ with $i<j$. It first visits the vertices of~$V$, starting with~$v_i$ and~$v_j$. As~$V$ is a clique in~$H$, $\sigma_{\le n}$ is a partial MNS-ordering.
Next,~$z$ is visited. Since its visited-neighbor set is~$S$ and no unvisited vertex has~$S$ as a subset of its visited-neighbor set,~$z$ has an inclusion-maximal visited-neighbor set. Thus~$\sigma_{\le n+1}$ is a partial MNS-ordering of~$H$.

The ordering then visits the vertices in~$T$. Each vertex~$t_x\in T$ has exactly two visited neighbors, namely the endpoints of the edge corresponding to~$t_x$, whereas every remaining unvisited vertex has at most one visited neighbor. Since~$T$ is an independent set, the ordering remains a partial MNS-ordering throughout this stage.

Next, the vertices from the clique~$U\setminus\{z\}$ are visited one after another. Each vertex of~$U\setminus\{z\}$ is adjacent to all previously visited vertices in~$T\cup\{z\}$. In contrast, no remaining unvisited vertex is adjacent to all vertices in~$T\cup\{z\}$. Hence the vertices in~$U\setminus\{z\}$ have inclusion-maximal visited-neighbor sets at this point. Since~$U\setminus\{z\}$ is a clique, the ordering remains a partial MNS-ordering throughout their visit.

Subsequently, the vertices from $F$ are visited, followed by those from~$W$.
By an analogous argument, the ordering continues to be a partial MNS-ordering of~$H$ throughout these visits.

    \item Orderings in group~(ii).

Consider the ordering~$\sigma$ constructed for a pair $\{i, j\} \subseteq [3n]$ with~$i < j$. It first visits the vertices of the clique~$U$, followed by the vertices of the independent set~$T$, each of which is adjacent to all vertices in~$U$. Consequently,~$\sigma$ remains a partial MNS-ordering throughout the visit of these vertices.

Next,~$\sigma$ visits the vertices of the independent set~$F$. Each vertex of~$F$ is adjacent to all vertices in~$U\setminus\{z\}$, whereas no remaining unvisited vertex is. Hence~$\sigma$ remains a partial MNS-ordering throughout this stage.

After that,~$\sigma$ visits the vertices of the clique~$V$. Each vertex~$v_x\in V$ has the visited neighbor~$f_x$, whereas no vertex in~$W$ is adjacent to~$f_x$. Hence the vertices of~$V$ have inclusion-maximal visited-neighbor sets and can be visited in the prescribed order. Finally,~$\sigma$ visits the vertices of the independent set~$W$, each of which has exactly two visited neighbors. Therefore,~$\sigma$ remains a partial MNS-ordering throughout this stage and is thus an MNS-ordering of~$H$.

    \item Orderings in group~(iii).

Consider the ordering~$\sigma$ constructed for some $i\in [m]$ and $j\in [3n]$. It first visits the vertices in~$\{t_i\}\cup U$, which induce a clique in~$H$. Hence~$\sigma$ remains a partial MNS-ordering throughout this stage.
The remaining vertices are visited in the same relative order as in the ordering of Group~(ii), and the same argument shows that~$\sigma$ is an MNS-ordering of~$H$.

    \item Orderings in group~(iv).

Consider the ordering~$\sigma$ constructed for some $i\in[n]$ and $j\in[3n-1]$. It first visits the~$3n$ vertices in~$\{f_i\}\cup(U\setminus\{z\})$, which induce a clique in~$H$. Hence~$\sigma_{\le 3n}$ is a partial MNS-ordering.
Next,~$z$ is visited. Since~$z$ is adjacent to all previously visited vertices except~$f_i$, whereas no remaining unvisited vertex is adjacent to all previously visited vertices,~$\sigma_{\le 3n+1}$ is also a partial MNS-ordering.
The remaining vertices are visited in the same relative order as in the ordering of Group~(ii). An argument analogous to that used there shows that~$\sigma$ is an MNS-ordering of~$H$.

    \item Orderings in group~(v).

Consider the ordering~$\sigma$ constructed for an integer~$i\in[n]$. It first visits~$v_i$ and~$f_i$, which are adjacent in~$H$. Therefore,~$\sigma_{\le 2}$ is a partial MNS-ordering.

Next,~$\sigma$ visits the vertices in the clique~$U\setminus\{z\}$, each of which is adjacent to~$f_i$. Since~$f_i$ is not adjacent to any of the remaining unvisited vertices, the vertices in~$U\setminus\{z\}$ have inclusion-maximal visited-neighbor sets. Hence~$\sigma_{\le 3n+1}$ remains a partial MNS-ordering.

Afterward,~$\sigma$ visits the vertices in the independent set~$T(i)$. By construction, each vertex in~$T(i)$ is adjacent to all vertices in~$\{v_i\}\cup(U\setminus\{z\})$, that is, to all previously visited vertices except~$f_i$. Since no remaining unvisited vertex is adjacent to all these vertices,~$\sigma$ remains a partial MNS-ordering throughout their visit.

Next,~$z$ is visited. Since~$z$ is adjacent to all previously visited vertices except~$f_i$, whereas no unvisited vertex is,~$\sigma$ remains a partial MNS-ordering.

The remainder of the ordering follows the same pattern as in Group~(i). An argument analogous to that used there shows that~$\sigma$ is an MNS-ordering of~$H$.

    \item Orderings in the group (vi).

Let $e_x = \edge{v_i}{v_j}$ be an edge of~$G$ with $i < j$.
We now show that the six orderings constructed for this edge in the group
are MNS-orderings of~$H$.

Consider the first ordering~$\sigma$ with the prefix
$(w_x, t_x, v_i, v_j, z)$.
The vertices~$w_x$,~$t_x$, and~$v_i$ form a clique in~$H$. Since~$v_j$ is adjacent to~$t_x$ and~$v_i$, and~$w_x$ has exactly two neighbors, namely~$v_i$ and~$t_x$, in~$H$, it follows that~$\sigma_{\le 4}$ is a partial MNS-ordering.

After~$v_j$, the ordering visits~$z\in U$.
Since~$S$ is a vertex cover,~$z$ is adjacent to at least one of~$v_i$ and~$v_j$. Moreover,~$z$ is adjacent to~$t_x$, and no remaining unvisited vertex is adjacent to~$w_x$. Therefore, any unvisited vertex whose visited-neighbor set strictly contains that of~$z$ would have to be adjacent to~$t_x$,~$v_i$, and~$v_j$. However, no remaining vertex has this property: vertices in~$V\setminus\{v_i,v_j\}$ are adjacent to~$v_i$ and~$v_j$ but not to~$t_x$, whereas all other remaining vertices are not adjacent to both~$v_i$ and~$v_j$. Hence~$z$ has an inclusion-maximal visited-neighbor set, and~$\sigma_{\le5}$ is a partial MNS-ordering.

The remaining vertices are visited according to
$(\Omega,\overrightarrow{W}\ominus\{w_x\})$.
Arguments analogous to those used in the previous groups show that each subsequent step preserves the partial MNS property. Hence~$\sigma$ is an MNS-ordering of~$H$.

The second, third, and fourth orderings are handled analogously by symmetry.

Consider now the fifth ordering~$\sigma$. It first visits~$t_x$,~$v_i$, and~$v_j$. Since these vertices form a clique in~$H$,~$\sigma_{\le 3}$ is a partial MNS-ordering. Next,~$z$ is visited. Since~$z$ is adjacent to~$t_x$ and to at least one of~$v_i$ and~$v_j$, whereas no unvisited vertex is adjacent to all three previously visited vertices,~$z$ has an inclusion-maximal visited-neighbor set. Hence~$\sigma_{\le 4}$ is a partial MNS-ordering.
The remaining vertices are visited as in the first ordering. An argument analogous to the previous ones shows that~$\sigma$ is an MNS-ordering of~$H$.
The sixth ordering is handled analogously by symmetry.
\end{itemize}

Thus,~$H$ is an MNS-support of~$\Pi$. Since~$\abs{\eset{H}}=\kappa+\abs{\eset{G'}}=k$, it follows that~$g(I)$ is a {\yesins}.
	
$(\Leftarrow)$ Assume that~$g(I)$ is a {\yesins}, and let~$H$ be an MNS-support of~$\Pi$ with at most~$k$ edges.
From the first two vertices of the constructed orderings, it follows that~$\eset{G'}\subseteq\eset{H}$.
Let~$e_x=\edge{v_i}{v_j}$ be any edge of~$G$, where $x\in[m]$ and $1\leq i<j\leq n$.
Consider the fifth ordering in Group~(vi).
After~$t_x$,~$v_i$, and~$v_j$ have been visited, the next vertex is~$z$.
Since~$w_x$ is adjacent to~$t_x$ and~$v_i$, the vertex~$z$ can have an inclusion-maximal visited-neighbor set only if it is adjacent to at least one of~$v_i$ and~$v_j$.
Hence~$z$ is adjacent in~$H$ to at least one endpoint of every edge of~$G$.
Let~$S=\neighbor{z}{H}\cap V$.
Then~$S$ is a vertex cover of~$G$.
Moreover, as~$z$ has no neighbors in~$V$ in~$G'$, every vertex in~$S$ contributes a distinct edge in~$\eset{H}\setminus\eset{G'}$.
Therefore, $\abs{S}\le \abs{\eset{H}}-\abs{\eset{G'}}\le k-\abs{\eset{G'}}=\kappa$ and $I$ is a {\yesins}.
\end{proof}

By introducing multiple pendant vertices attached to the vertex~$z$ in the above reduction, and adding one additional ordering for each $i \in [n]$, we obtain the following result.

\begin{theorem}
\label{thm-mns-degree-npc}
\prob{Deg-Bounded-MNS-SP} is {\nph}.
\end{theorem}

\begin{proof}
We modify the reduction from the proof of Theorem~\ref{thm-mns-edges-npc}. Introduce~$n$ new vertices~$z_1,z_2,\ldots,z_n$, and let
$\overrightarrow{Z}=(z_1,z_2,\ldots,z_n)$.
Append~$\overrightarrow{Z}$ to every ordering in~$\Pi$. In addition, for each~$i\in[n]$, add the ordering
$
\left(
z_i,
z,
\overrightarrow{Z}\ominus\{z_i\},
\overrightarrow{U}\ominus\{z\},
\overrightarrow{T},
\overrightarrow{F},
\overrightarrow{V},
\overrightarrow{W}
\right)$.
The new orderings force $z$ to be adjacent to the $n$ vertices in~$Z$ in any MNS-support. Recall that~$z$ is adjacent in~$G'$ to all~$m+3n-1$ vertices in~$T\cup(U\setminus\{z\})$.
The remainder of the argument is analogous to that in the proof of Theorem~\ref{thm-mns-edges-npc}: additional neighbors of~$z$ in~$V$ correspond exactly to a vertex cover of~$G$. Hence~$G$ has a vertex cover of size at most~$\kappa$ if and only if there exists an MNS-support of maximum degree  at most~$\kappa+m+4n-1$.
\end{proof}

\section{A Polynomial-time Algorithm for Recognizing DFS-Tree Supports}
This section presents a polynomial-time algorithm for {\prob{SP-DFS-Tree}}.
Throughout, let $I=(V,\Pi)$ denote an instance of {\prob{SP-DFS-Tree}}, with $m=\abs{\Pi}$ and $n=\abs{V}$.

\paragraph{Algorithm outline.}
Our algorithm first computes an attachment graph that contains every DFS-tree support of~$\Pi$ as a subgraph whenever one exists, building on prior work by Peters~et~al.~\cite{PetersYCE22}.
We then identify a subgraph~$T^{\star}$ of the attachment graph that is a tree, and show that, whenever~$I$ is a {\yesins}, there exists a DFS-tree support that contains~$T^{\star}$ and attaches all remaining vertices as leaves.

To build intuition for our algorithm, we recall the standard recursive view of DFS on rooted trees.

\begin{definition}[Recursive description of DFS on trees]
\label{def-recursive-dfs}
Given a rooted tree $T$ with root $r$, DFS proceeds as follows:
\begin{enumerate}[(1)]
    \item Visit the root $r$.
    \item Recursively perform DFS on each child of $r$, processing one subtree at a time until all children have been visited.
    \item After finishing with one child and its entire subtree, backtrack to~$r$ and continue with the next (unvisited) child.
\end{enumerate}
\end{definition}
In general, DFS explores each subtree as deeply as possible before moving on to the next, producing an ordering in which all vertices of a subtree appear consecutively.

When~$\Pi$ admits a DFS-tree support~$T$, the tree~$T$ is unrooted. Nevertheless, when interpreting a DFS traversal of~$T$ with respect to an ordering~$\sigma$, we treat~$T$ as rooted at the first vertex of~$\sigma$. In particular, the recursive structure of DFS imposes strong constraints on the relative positions of vertices in a DFS-ordering, which we exploit throughout the remainder of the section.


\begin{observation}
\label{obs-dfs-u-v-path-r-u}
Let~$\sigma$ be a DFS-ordering of a tree~$T$, let~$r$ be the first vertex in~$\sigma$, and consider~$T$ rooted at~$r$.
For any vertices~$u,v\in V(T)$ with~$u <_{\sigma} v$, if every vertex between~$u$ and~$v$ in~$\sigma$ is a leaf of~$T$, then the parent of~$v$ lies on the unique from $r$ to $u$ in~$T$.
\end{observation}

\begin{observation}
\label{obs-xyz}
Let~$\sigma$ be a DFS-ordering of a tree~$T$, let~$r$ be the first vertex in~$\sigma$, and consider~$T$ rooted at~$r$.
For any vertices~$x,y,z\in \vset{T}$ such that
\begin{enumerate}[(1)]
    \item the parents of~$x$ and~$y$ are distinct and both lie on the path from~$r$ to~$z$ in~$T$, and
    \item $z <_{\sigma} y <_{\sigma} x$,
\end{enumerate}
the parent of~$y$ lies closer to~$z$ on the~$r$-$z$ path than the parent of~$x$.
\end{observation}
%


\subsection{Attachment Digraphs}
Trick~\cite{MichaeTricksinglepeaked89} presented a polynomial-time algorithm for~\prob{SP-GS-Tree}. Building on this work, Peters~et~al.~\cite{PetersYCE22} introduced the notion of an \emph{attachment digraph} and obtained a complete characterization of the families~$\Pi$ that admit a GS-tree support. A key property of the attachment digraph is that every GS-tree support of~$\Pi$ is a subgraph of its underlying graph.

For an ordering~$\sigma$ of~$V$ and a subset~$S\subseteq V$, let
$\tp{\sigma}{S}$, $\second{\sigma}{S}$, and $\bottom{\sigma}{S}$
denote the first, second, and last vertices of~$S$ in~$\sigma$,
respectively. When~$S=V$, we omit~$S$ from the notation; for example,
$\ttp{\sigma}=\tp{\sigma}{V}$.
For~$x\in S$, define
\[B(\sigma, S, x) =
\begin{cases}
    \{\second{\sigma}{S}\},& \text{if}~\tp{\sigma}{S}= x, \\
    \{y \in S \setmid y <_\sigma x \},& \text{otherwise}.\\
\end{cases}
\]

We now introduce Algorithm~\ref{alg-att-digraph}, which takes as input a set~$\Pi$ of orderings on a vertex set~$V$.
The algorithm operates iteratively. At each step, it removes all vertices that are bottom-ranked in the current working set~$S$ in at least one ordering of~$\Pi$, where initially~$S=V$.
These vertices form the set~$L$.
For each vertex $v \in L$, the algorithm computes
\[
  B(v) = \bigcap_{\sigma \in \Pi} B(\sigma, S, v),
\]
and introduces arcs from~$v$ to every vertex in~$B(v)$.
Prior work~\cite{PetersYCE22,MichaeTricksinglepeaked89} shows that every neighbor of~$v$ 
in a GS-tree support of~$\Pi$ must belong to~$B(v)$.
After creating these arcs, the vertices in~$L$ are removed from~$S$, and the process repeats.
The iteration continues until $S$ contains only two vertices, at which point the algorithm adds a final arc between them in an arbitrary direction.
The resulting digraph is referred to as an attachment digraph of~$\Pi$. Since the only nondeterministic step of the algorithm is the orientation of the final arc, each instance~$\Pi$ gives rise to at most two attachment digraphs.
%

\begin{algorithm}[ht]
\caption{}
\label{alg-att-digraph}
\prealg
{A set $\{\sigma_i \setmid i \in [t]\}$ of~$t$ orderings on a vertex set~$V$}
{A digraph $(V, A)$}
\begin{algorithmic}[1]
    \State $A \gets \emptyset$; \Comment{Initialize the arc set}
    \State $S \gets V$; \Comment{Initialize the working set}
    \While{$\abs{S} \ge 3$}
        \State $L \gets \{\bottom{\sigma_i}{S} \setmid i \in [t]\}$ \Comment{Select bottom-ranked elements}
        \For{each $v \in L$} \label{alg-attg-for-loop-start} \label{alg-att-L}
            \State $B(v) \gets \bigcap_{i \in [t]} B(\sigma_i, S, v)$; \Comment{Common blockers of $v$}
            \If{$B(v) \neq \emptyset$} \label{alg-attg-Bv-nonempty}
                \For{each $u \in B(v)$} \label{alg-att-precondition-add-arc}
                    \State add arc $(v, u)$ to $A$; \label{al-attg-add-arc}
                \EndFor
            \EndIf
        \EndFor
        \State $S \gets S \setminus L$; \label{alg-attg-update-S}
    \EndWhile
    \If{$\abs{S} = 2$}
        \State add an arc in an arbitrary direction between the two vertices of~$S$ to~$A$;
    \EndIf
    \State \Return $(V, A)$;
\end{algorithmic}
\end{algorithm}

Let~$\attg{\Pi}$ denote an 
attachment digraph of~$\Pi$, and let~$\attug{\Pi}$ denote its underlying undirected graph, called the attachment graph of~$\Pi$. The attachment graph is uniquely determined by~$\Pi$ and can be computed in polynomial time using Algorithm~\ref{alg-att-digraph}. Moreover, whenever~$\Pi$ admits a GS-tree support,~$\attug{\Pi}$ contains every GS-tree support of~$\Pi$ as a subgraph~\cite{PetersYCE22}.

The following characterization due to Peters~et~al.~\cite{PetersYCE22} links the existence of GS-tree supports to the connectivity of the attachment graph.

\begin{lemma}[\cite{PetersYCE22}]
\label{lem-characterization-gs}
The set $\Pi$ admits a GS-tree support if and only if $\attug{\Pi}$ is connected.
\end{lemma}



The connectivity characterization can be strengthened as follows.

\begin{lemma}[\protect{\cite[Theorem~7.9]{PetersYCE22}}]
\label{lem-att-dig-realize}
A tree $T$ is a GS-tree support of $\Pi$ if and only if it can be obtained from $\attg{\Pi}$ by
retaining exactly one outgoing arc from each nonsink vertex,
and then replacing all remaining directed arcs by undirected edges.
\end{lemma}



A vertex in an attachment digraph is called a \emph{forced vertex} if it has at most one outneighbor, and a \emph{free vertex} otherwise.
As observed in~\cite{PetersYCE22}, every attachment digraph contains at least one forced vertex. In particular, the first vertex of every ordering is forced.

\begin{lemma}[\cite{PetersYCE22}]
\label{lem:1st-vertex-forced}
Then the first vertex of every ordering in~$\Pi$ is a forced vertex of~$\attg{\Pi}$.
\end{lemma}
%
%
%
%

The next lemma describes the structure of the forced vertices in instances that admit a GS-tree support.
Let~$T^{\star}$ denote the subgraph of~$\attug{\Pi}$ induced by the forced vertices of~$\attg{\Pi}$.

\begin{lemma}[\cite{PetersYCE22}]
\label{lem-forced-vertex}
If~$\Pi$ admits a GS-tree support, then~$T^{\star}$ is a tree.
\end{lemma}


As established above, every GS-tree support of~$\Pi$ is a subgraph of~$\attug{\Pi}$ whenever such a support exists.
Since every DFS-ordering is also a GS-ordering, every DFS-tree support of~$\Pi$ is a GS-tree support of~$\Pi$.
Together with Lemma~\ref{lem-forced-vertex}, this yields the following corollary.
%
%

\begin{corollary}
\label{cor-subtree-attachmentgraph}
If~$\dfssu{\Pi}\neq\emptyset$, then every tree in~$\dfssu{\Pi}$ is a subgraph of~$\attug{\Pi}$ and contains~$T^{\star}$ as a subtree.
\end{corollary}

\subsection{The Polynomial-time Algorithm}
We now present a polynomial-time algorithm for \prob{SP-DFS-Tree}.
The problem is trivial when $\Pi$ consists of a single ordering.
Thus, we assume throughout that $\Pi$ contains at least two orderings.

As a first step, we construct an attachment digraph~$\attg{\Pi}$ of~$\Pi$ using Algorithm~\ref{alg-att-digraph}.
For convenience, let $D = \attg{\Pi}$ and $G = \attug{\Pi}$.

If~$G$ is disconnected, by Lemma~\ref{lem-characterization-gs},~$\Pi$ admits no GS-tree support, and therefore no DFS-tree support.
We therefore return ``{\no}''. Assume henceforth that~$G$ is connected.

Let~$V^{\star}$ denote the vertex set of~$T^{\star}$, i.e., the set of forced vertices.
By Lemma~\ref{lem-forced-vertex},~$T^{\star}$ is a tree. Moreover, by Lemma~\ref{lem:1st-vertex-forced},~$T^{\star}$ contains the first vertex~$\ttp{\sigma}$ of every ordering~$\sigma\in\Pi$.
By Corollary~\ref{cor-subtree-attachmentgraph}, if $I$ is a {\yesins}, then $T^{\star}$ is a subtree of every DFS-tree support of~$\Pi$.

The remaining task is to determine whether~$T^{\star}$ can be extended to a DFS-tree support of~$\Pi$. The next lemma provides a crucial structural simplification: whenever a DFS-tree support exists, there is one in which every vertex outside~$V^{\star}$ is attached to~$T^{\star}$ as a leaf.

\begin{lemma}
\label{lem-dfs-leaves}
If $\dfssu{\Pi}\neq\emptyset$,
then there exists~$T\in \dfssu{\Pi}$ such that
\begin{itemize}
    \item $T^{\star}$ is a subtree of~$T$, and
    \item every vertex in $V \setminus V^{\star}$ is a leaf of~$T$.
\end{itemize}
\end{lemma}

\begin{proof}
Suppose that $\dfssu{\Pi}\neq\emptyset$.
Let~$T'$ be any arbitrary tree in $\dfssu{\Pi}$. By Corollary~\ref{cor-subtree-attachmentgraph},~$T^{\star}$ is a subtree of~$T'$.
Let~$T_1$,~$T_2$,~$\ldots$,~$T_k$ denote the connected components of $T' - {V^{\star}}$. For each $i \in [k]$, let~$u_i$ be the vertex in~$T_i$ that is adjacent in~$T'$ to some vertex $v_i \in {V^{\star}}$.

Let~$\sigma$ be an ordering in~$\Pi$, which is a DFS-ordering of~$T'$. By Lemma~\ref{lem:1st-vertex-forced},~$T^{\star}$ contains~$\ttp{\sigma}$. Consequently, when performing a DFS traversal of~$T'$ according to~$\sigma$, all vertices of each~$T_i$, $i \in [k]$, are visited consecutively, beginning with~$u_i$.
Equivalently, in the ordering~$\sigma$, for every~$i \in [k]$, the vertices of~$T_i$ appear as a consecutive block, with~$u_i$ as the first element.

Construct a tree~$T$ by extending~$T^{\star}$: for each~$i \in [k]$, attach every vertex of~$T_i$ as a leaf adjacent to~$v_i$.
%
%
It is easy to see that~$\sigma$ remains a DFS-ordering of~$T$.
Since this holds for all~$\sigma \in \Pi$,~$T$ is a DFS-tree support of~$\Pi$. The lemma follows.
\end{proof}

By Lemma~\ref{lem-dfs-leaves}, if~$\Pi$ admits a DFS-tree support, then~$T^{\star}$ must be a DFS-tree support of~$\Pi$ restricted to the forced vertices. This yields the following necessary condition.

For $V' \subseteq V$, let~$\Pi_{|V'}$ denote the set of orderings obtained by restricting each ordering in~$\Pi$ to~$V'$.

\begin{observation}
\label{obs-t-star-fail}
If~$T^{\star}$ is not a DFS-tree support of~$\Pi_{|V^{\star}}$, then~$\Pi$ admits no DFS-tree support.
\end{observation}
%
By Observation~\ref{obs-t-star-fail}, if~$T^{\star}$ is not a DFS-tree support of~$\Pi_{|V^{\star}}$, we return ``\no''. Henceforth, we assume that~$T^{\star}$ is a DFS-tree support of~$\Pi_{|V^{\star}}$.
By Lemma~\ref{lem-dfs-leaves}, it suffices to determine, for each vertex in~$V \setminus V^{\star}$, a suitable parent in~$T^{\star}$. To facilitate the discussion, we introduce the following notation.

Let~$\sigma$ be an ordering in~$\Pi$. For a {\free} vertex~$v$, its \emph{guider} is the closest {\forced} vertex preceding~$v$ in~$\sigma$. We denote it by~$\guide{v}{\sigma}$. Since the first vertex of~$\sigma$ is {\forced} (Lemma~\ref{lem:1st-vertex-forced}), every {\free} vertex in~$\sigma$ has a guider.
For two {\forced} vertices~$u$ and~$v$, let~$P_{\{u, v\}}$ denote the unique path between~$u$ and~$v$ in~$T^{\star}$.
For a {\forced} vertex~$v$ other than the first in~$\sigma$, let~$\parent{\sigma}{v}$ denote its parent in the tree~$T^{\star}$ rooted at the first vertex of~$\sigma$.
If~$\sigma$ is a DFS-ordering of a tree, then every vertex $v \in V^{\star}$, except the first one in~$\sigma$, is adjacent to exactly one vertex that precedes~$v$ in~$\sigma$, namely~$\parent{\sigma}{v}$.

Let~$\mathcal{T}^{\star}$ denote the set of all trees on~$V$ that contain~$T^{\star}$ as a subtree and in which every {\free} vertex is a leaf.

The following lemma restricts the possible neighbors of a {\free} vertex.


\begin{lemma}
\label{lem:undeter-vertex}
Let~$\sigma$ be an ordering of~$V$, let~$v$ be a {\free} vertex, and let~$u = \guide{v}{\sigma}$.
Let~$u'$ denote the closest {\forced} vertex succeeding~$v$ in~$\sigma$, if one exists, and define
\[
w =
\begin{cases}
\ttp{\sigma}, & \text{if no {\forced} vertex succeeds~$v$ in~$\sigma$},\\
\parent{\sigma}{u'}, & \text{otherwise}.
\end{cases}
\]
If there exists a tree~$\widehat{T} \in \mathcal{T}^{\star}$ admitting~$\sigma$ as a DFS-ordering, then the unique neighbor of~$v$ in~$\widehat{T}$ lies on~$P_{\{u, w\}}$.
\end{lemma}

\begin{proof}
Let~$\widehat{T} \in \mathcal{T}^{\star}$ be a tree admitting~$\sigma$ as a DFS-ordering. Root~$\widehat{T}$ at the first vertex of~$\sigma$, and consider the corresponding DFS traversal.
We distinguish two cases.

\medskip
\noindent\textbf{Case~1:} No {\forced} vertex succeeds~$v$ in~$\sigma$.

In this case, $w$ appears first in~$\sigma$ and is therefore a {\forced} vertex by Lemma~\ref{lem:1st-vertex-forced}. Moreover, all vertices succeeding~$u$ in~$\sigma$ are leaves of~$\widehat{T}$.
After~$u$, the traversal backtracks along~$P_{\{u,w\}}$ toward~$w$, exploring the unvisited children of vertices on this path.
Consequently, every {\free} vertex succeeding~$u$ in~$\sigma$, including~$v$, is a child of some vertex on~$P_{\{u,w\}}$.

\medskip
\noindent\textbf{Case~2:} A {\forced} vertex succeeds~$v$ in~$\sigma$.

Since~$u'$ is a {\forced} vertex and every {\free} vertex is a leaf of~$\widehat{T}$, the parent of~$u'$ in~$\widehat{T}$ is also {\forced}. As~$T^{\star}$ contains all determined vertices and is a subtree of~$\widehat{T}$, this parent is~$w=\parent{\sigma}{u'}$.
By Observation~\ref{obs-dfs-u-v-path-r-u},~$w$ lies on the path from the root to~$u$ in~$\widehat{T}$.
Since~$\sigma$ is a DFS-ordering of~$\widehat{T}$, after visiting~$u$ and before visiting~$u'$, the traversal backtracks along~$P_{\{u,w\}}$ and explores unvisited children of vertices on this path.
Every vertex between~$u$ and~$u'$ in~$\sigma$ is {\free}, and hence a leaf of~$\widehat{T}$.
 Therefore, every such vertex, including~$v$, is a child of some vertex on~$P_{\{u,w\}}$.
%
%
\end{proof}

\begin{figure}
    \centering
    \includegraphics[width=0.75\linewidth]{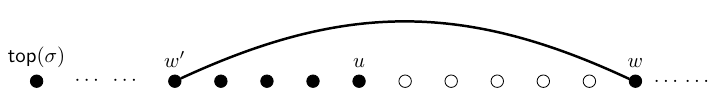}
    \caption{An illustration of the ordering~$\sigma$. Forced vertices are shown in black and {\free} vertices in white. The arc represents the edge in~$T^{\star}$ between the forced vertices~$w'$ and~$w$. The vertex~$u$ is the guider, in~$\sigma$, of all white vertices lying between~$u$ and~$w$.}
    \label{fig-dfs-tree}
\end{figure}




By Lemma~\ref{lem:undeter-vertex}, for every ordering~$\sigma$ and every {\free} vertex~$v$, the unique neighbor of~$v$ in any DFS-tree support of~$\Pi$ within~$\mathcal{T}^{\star}$ must lie on a path~$P_{\{u,w\}}$ in~$T^{\star}$, where~$u$ and~$w$ are as specified in the lemma.
Let~$Q(\sigma,v)$ denote the vertex set of~$P_{\{u,w\}}$. Define
\begin{equation}
\label{eq-qv}
    Q(v) = \bigcap_{\sigma \in \Pi} Q(\sigma, v).
\end{equation}
Then, if~$\Pi$ admits a DFS-tree support~$\widehat{T}\in \mathcal{T}^{\star}$, the unique neighbor of every {\free} vertex~$v$ in~$\widehat{T}$ belongs to~$Q(v)$.
Since each~$P_{\{u, w\}}$ is an induced path in the tree~$T^{\star}$, it follows that $Q(v)$, being the intersection of these paths, is itself a subpath of every such~$P_{\{u, w\}}$. Consequently,~$Q(v)$ induces a path in~$T^{\star}$.
We formalize this observation as follows.

\begin{observation}
\label{lem-qv-subpath}
Let~$v$ be a {\free} vertex, let~$\sigma$ be an ordering in~$\Pi$, let $u = \guide{v}{\sigma}$, and let~$w$ be the first vertex of~$\sigma$. If there exists a tree $\widehat{T} \in \mathcal{T}^{\star}$ admitting $\sigma$ as a DFS-ordering, then $Q(v)$ induces a subpath of~$P_{\{u, w\}}$ in~$T^{\star}$.
\end{observation}


We next introduce a purification rule that further refines~$Q(v)$ for each {\free} vertex~$v$. Before presenting this rule, we establish a lemma that justifies it.

For two {\free} vertices~$a$ and~$b$, and a {\forced} vertex~$c$, we define:
\begin{itemize}
    \item $Q_{(b,c)}^{>}(a)$ as the set of vertices in~$Q(a)$ that are farther from~$c$ in~$T^{\star}$ than every vertex in~$Q(b)$, and
    \item $Q_{(b,c)}^{<}(a)$ as the set of vertices in~$Q(a)$ that are closer to~$c$ in~$T^{\star}$ than every vertex in~$Q(b)$.
\end{itemize}

\begin{lemma}
\label{lem-purify}
Let~$\sigma$ be an ordering on~$V$, and let~$u$ and~$v$ be two {\free} vertices that share the same guider~$w$ in~$\sigma$, with $u <_\sigma v$.
If there is a tree $\widehat{T}\in \mathcal{T}^{\star}$ that admits $\sigma$ as a DFS-ordering, then the unique neighbor of $v$ in $\widehat{T}$ is from $Q(v) \setminus Q_{(u,w)}^{<}(v)$, and the unique neighbor of $u$ in $\widehat{T}$ is from $Q(u) \setminus Q_{(v,w)}^{>}(u)$.
\end{lemma}

\begin{proof}
   Assume that there exists a tree $\widehat{T} \in \mathcal{T}^{\star}$ that admits~$\sigma$ as a DFS-ordering. For brevity, let~$\widehat{N}(v)$  (respectively,~$\widehat{N}(u)$) denote the unique neighbor of $v$ (respectively,~$u$) in~$\widehat{T}$. By Lemma~\ref{lem:undeter-vertex}, we have $\widehat{N}(v)\in Q(v)$ and $\widehat{N}(u)\in Q(u)$.

  Let~$r$ be the first vertex of~$\sigma$, and root~$\widehat{T}$ at~$r$. Since~$\widehat{T}$ contains~$T^{\star}$ as a subtree, and all free vertices are leaves, $\widehat{N}(v)$ and $\widehat{N}(u)$ are precisely the parents of~$v$ and~$u$ in~$\widehat{T}$, respectively.

  Since~$u$ and~$v$ are {\free} vertices with the same guider~$w$ in~$\sigma$, and~$\sigma$ is a DFS-ordering of $\widehat{T}$, Observation~\ref{lem-qv-subpath} implies that both~$Q(v)$ and~$Q(u)$ induce subpaths of the path from~$r$ to~$w$ in~$\widehat{T}$.
  Therefore, both~$\widehat{N}(v)$ and~$\widehat{N}(u)$ lie on this path.
  Consider the DFS traversal of~$\widehat{T}$ with respect to~$\sigma$. After visiting $w$, the traversal backtracks toward~$r$. Since~$u$ is visited before~$v$ in $\sigma$, either $\widehat{N}(u) = \widehat{N}(v)$, or the traversal backtracks to $\widehat{N}(u)$ before $\widehat{N}(v)$, implying that $\widehat{N}(u)$ is closer to $w$ than $\widehat{N}(v)$ in $T^{\star}$.
   It follows that $\widehat{N}(v) \in Q(v) \setminus Q_{(u,w)}^{<}(v)$. Similarly, $\widehat{N}(u) \in Q(u) \setminus  Q_{(v,w)}^{>}(u)$.
\end{proof}

We now introduce the following purification rule and apply it recursively to shrink~$Q(v)$ for every {\free} vertex~$v$.

\begin{purificationrule}
\label{purification-rule}
If there exist two {\free} vertices $u$ and $v$ and an ordering $\sigma \in \Pi$ such that:
\begin{enumerate}[(i)]
    \item $u <_\sigma v$,
    \item $u$ and $v$ share the same guider $w$ in $\sigma$, and
    \item
    $
    Q_{(u,w)}^{<}(v) \cup Q_{(v,w)}^{>}(u) \neq \emptyset.
    $
\end{enumerate}
Then update $Q(v) := Q(v) \setminus Q_{(u,w)}^{<}(v)$ and $Q(u) := Q(u) \setminus Q_{(v,w)}^{>}(u)$.
\end{purificationrule}

We exhaustively apply the purification rule, in arbitrary order, until it no longer applies. Each application removes at least one vertex from~$Q(u)$ for some {\free} vertex~$u$. Moreover, each application can be performed in polynomial time. Since there are at most~$n-1$ {\free} vertices, each satisfying~$\abs{Q(u)}\le n-1$, the rule can be applied only polynomially many times. Hence, the purification rule can be exhaustively applied in polynomial time.

After exhaustively applying the purification rule, if there exists a {\free} vertex~$v$ with~$Q(v)=\emptyset$, then, by Lemma~\ref{lem-purify}, the instance~$I$ is a {\noins}.

We now arrive at our final task: assigning each {\free} vertex to a {\forced} vertex.

\begin{description}
    \item[Parent assignment phase:]
    We first fix an arbitrary bijection $f: V \to [n]$, referred to as the \emph{assignment function}.
    Then, for each {\free} vertex~$v$, we add~$v$ to~$T^{\star}$ as a leaf adjacent to
    a vertex $u \in Q(v)$ satisfying
    $f(u) \leq f(u')$ for all $u' \in Q(v)$.
\end{description}

The parent assignment phase can be carried out in polynomial time. We next establish its correctness through the following lemmas. We emphasize that, upon entering the parent assignment phase, the restriction of every ordering in~$\Pi$ to the {\forced} vertices is a DFS ordering of~$T^{\star}$.

\begin{lemma}
\label{lem-differ-stage}
Let~$u$ and~$v$ be two {\free} vertices, and let~$\sigma$ be an ordering of~$V$ such that $u <_\sigma v$, and~$u$ and~$v$ have different guiders in~$\sigma$. Let $u' = \guide{u}{\sigma}$. If there exists a tree~$\widehat{T}$ in~$\mathcal{T}^{\star}$ that admits~$\sigma$ as a DFS-ordering, then the following statements hold:
\begin{enumerate}[(1)]
    \item $\abs{Q(\sigma, u) \cap Q(\sigma, v)} \le 1$.

    \item Every vertex in $Q(\sigma, u) \setminus Q(\sigma, v)$ is closer to~$u'$ in $T^{\star}$ than every vertex in~$Q(\sigma, v)$.

    \item Every vertex in $Q(\sigma, u) \cap Q(\sigma, v)$ is closer to~$u'$ in~$T^{\star}$ than every vertex in $Q(\sigma, v) \setminus Q(\sigma, u)$.
\end{enumerate}
\end{lemma}

\begin{proof}
    Assume that there exists $\widehat{T} \in \mathcal{T}^{\star}$  admitting~$\sigma$ as a DFS-ordering.
    Let $v' = \guide{v}{\sigma}$ denote the guider of $v$ in $\sigma$, which is a {\forced} vertex.
    Since $u<_{\sigma} v$, we have $u'<_{\sigma} u<_{\sigma} v'<_{\sigma}  v$.  Let~$r$ be the first vertex in~$\sigma$, and root both~$T^{\star}$ and~$\widehat{T}$ at~$r$.
    Let~$w$ be the nearest {\forced} vertex to the right of~$u$ in~$\sigma$, so that either $w=v'$ or $w<_{\sigma} v'$. Let $w' = \parent{w}{\sigma}$ denote the parent of~$w$ in~$T^{\star}$.
 Figure~\ref{fig-dfs-tree-b} illustrates the relative positions of these vertices in~$\sigma$.

    \begin{figure}
        \centering
        \includegraphics[width=0.85\textwidth]{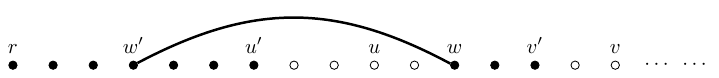}
        \caption{An illustration of the relative ordering of vertices in the ordering~$\sigma$ used in the proof of Lemma~\ref{lem-differ-stage}. Dark vertices represent {\forced} vertices, while the remaining vertices are {\free}. The edge between~$w$ and~$w'$ in~$T^{\star}$ is shown; all other edges are omitted for clarity.}
        \label{fig-dfs-tree-b}
    \end{figure}

We first establish Statement~(1). By definition, $Q(\sigma, u)$ is the set of vertices on the path $P_{\{u', w'\}}$.
If~$u'=w'$, then the statement is immediate.
Otherwise, let~$w''$ be the neighbor of~$w'$ on this path.
Consider the DFS traversal of~$\widehat{T}$ according to~$\sigma$.
Since~$\widehat{T}$ contains~$T^{\star}$ as a subtree and every {\free} vertex is a leaf, after visiting~$u'$ and before visiting~$w$, the traversal backtracks along the path~$P_{\{u', w'\}}$ from~$u'$ to~$w'$.
As neither~$v'$ nor~$v$ is visited before~$w$, neither is a descendant of any vertex in~$P_{\{u', w''\}}$ (Definition~\ref{def-recursive-dfs}). Hence,  either $P_{\{u', w'\}}$ and $P_{\{v', r\}}$ are vertex-disjoint, or they intersect only at~$w'$, the latter occurring if and only if~$v'$ is a descendant of~$w'$.
    By Observation~\ref{obs-dfs-u-v-path-r-u} and Lemma~\ref{lem:undeter-vertex}, $Q(\sigma, v)\subseteq \vset{{P_{\{v', r\}}}}$. Therefore,
    \[\abs{Q(\sigma, u) \cap Q(\sigma, v)} \le \abs{\vset{P_{\{u', w'\}}} \cap \vset{P_{\{v', r\}}}} \le 1.\]

We now prove Statements~(2) and~(3). As discussed, neither~$v'$ nor~$v$ is a descendant of any vertex on~$P_{\{u', w''\}}$ in~$\widehat{T}$. Consequently, for any vertex~$x \in Q(\sigma, v)$, the path from~$x$ to~$u'$ in~$\widehat{T}$ contains~$P_{\{w',u'\}}$ as a subpath. Since $Q(\sigma,u)$ is precisely the vertex set of~${P_{\{u', w'\}}}$, Statement~(2) follows immediately.  Moreover, we have shown that $Q(\sigma, u) \cap Q(\sigma, v) \in \{\emptyset, \{w'\}\}$, which establishes Statement~(3).
\end{proof}

For clarity, we use different notation for the tree before and after the assignment phase. Specifically, $T^{\star}$ denotes the tree before the assignment phase, while~$\widetilde{T}$ denotes the tree obtained afterward. Observe that~$\widetilde{T}\in\mathcal{T}^{\star}$.
 For each {\free} vertex~$x$, let~$\tilde{x}$ denote its unique neighbor in~$\widetilde{T}$. By construction,
$\tilde{x} = \arg\min_{y \in Q(x)} f(y)$, where~$f$ is the assignment function.

\begin{lemma}
\label{lem-correctness-assignment}
  If~$\Pi$ admits a DFS-tree support, then~$\widetilde{T}$ is a DFS-tree support of~$\Pi$.
\end{lemma}

\begin{proof}
   Assume that~$\Pi$ admits a DFS-tree support. Suppose for contradiction that there exists an ordering~$\sigma\in\Pi$ that is not a DFS-ordering of~$\widetilde{T}$. Let~$k$ be the smallest index such that~$\sigma_{\le k-1}$ is a partial DFS-ordering of~$\widetilde{T}$, but~$\sigma_{\le k}$ is not, and let~$u$ be the $k$-th vertex of~$\sigma$.
   Observe that~$u$ is a {\free} vertex.
    Let $r = \ttp{\sigma}$, and let $w = \guide{u}{\sigma}$.

    Let~$\sigma'$ be a DFS-ordering of~$\widetilde{T}$ containning $\sigma_{\leq k-1}$ as a prefix.
    Let~$v$ be the~$k$-th vertex of~$\sigma'$. By Lemma~\ref{lem-prefix}, such a DFS-ordering~$\sigma'$ exists. Clearly,~$v\neq u$.
   We next prove that~$v$ is also {\free}.

\begin{claim}
$v$ is a {\free} vertex.
\end{claim}

\begin{proof}
Root both~$T^{\star}$ and~$\widetilde{T}$ at~$r$.
Assume, for contradiction, that~$v$ is {\forced}.
Let~$v^*$ be the parent of~$v$ in~$T^{\star}$.
There are no {\forced} vertices between~$w$ and~$v$ in~$\sigma'$, and~$\sigma'$ restricted to the {\forced} vertices is a DFS-ordering of~$T^{\star}$.
By Observation~\ref{obs-dfs-u-v-path-r-u},~$v^*$ lies on the path between~$w$ and the root~$r$ in~$T^{\star}$.

Let~$a$ be the closest {\forced} vertex to the right of~$u$ in~$\sigma$; such a vertex exists since~$v$ is {\forced} and lies to the right of~$u$.
Let~$a^*$ be the parent of~$a$ in~$T^{\star}$.
Observe that all vertices between~$a$ and~$w$ in~$\sigma$ are {\free}.
Moreover, upon entering the parent assignment phase, every ordering in~$\Pi$, including~$\sigma$, restricted to the {\forced} vertices is a DFS-ordering of~$T^{\star}$.
Similarly, Observation~\ref{obs-dfs-u-v-path-r-u} implies that~$a^*$ lies on the path between~$w$ and~$r$ in~$T^{\star}$.

We claim that~$a^*$ lies on the path~$P_{\{v^*, w\}}$ in~$T^{\star}$.
If $a^*=v^*$, the claim is immediate.
Otherwise, by the choice of~$a$, we have $a <_{\sigma} v$.
Since~$\sigma$ restricted to the {\forced} vertices is a DFS-ordering of~$T^{\star}$, and $w <_{\sigma} a <_{\sigma} v$, Observation~\ref{obs-xyz} implies that~$a^*$ is closer to~$w$ than~$v^*$ in~$T^{\star}$. As both~$a^*$ and~$v^*$ lie on the path between~$w$ and~$r$ in~$T^{\star}$, it follows that~$a^*$ lies on~$P_{\{v^*, w\}}$, as claimed.

We now derive a contradiction.
Recall that~$\sigma'$ is a DFS-ordering of~$\widetilde{T}$.
Furthermore,~$u$ is a leaf in~$\widetilde{T}$ whose parent lies on the path between~$w$ and~$r$ (Lemma~\ref{lem:undeter-vertex}), we have
\[\tilde{u}\in Q(u)\subseteq Q(\sigma, u) = \vset{P_{\{a^*, w\}}},\]
and hence $\tilde{u}$ is on the path between~$w$ and~$a^*$ in~$\widetilde{T}$.
Since both~$v$ and~$u$ are visited after~$w$, and~$v$ is visited before~$u$ in~$\sigma'$, Observation~\ref{obs-xyz} implies that~$v^*$ lies on the path between~$\tilde{u}$ and~$w$.
This is possible only if $a^*=v^*=\tilde{u}$. That is,~$u$ and~$v$ have the same parent.
Replacing~$v$ with~$u$ in~$\sigma'_{\le k}$ therefore preserves the partial DFS-ordering of~$\widetilde{T}$ on the first~$k$ vertices.
Since $\sigma'_{\leq k-1}$ is the same as~$\sigma_{\leq k-1}$, this contradicts the minimality of~$k$.
%
\end{proof}

Recall that~$\tilde{u}$ and~$\tilde{v}$ denote the unique neighbors of~$u$ and~$v$ in~$\widetilde{T}$, respectively.
   Clearly,~$\tilde{u} \neq \tilde{v}$, since otherwise~$\sigma_{\leq k}$ would form a partial DFS-ordering of~$\widetilde{T}$, which contradicts our assumption.
We prove the following two statements.

    \begin{enumerate}[(i)]
        \item In~$T^{\star}$,~$\tilde{v}$ lies on the path between~$w$ and~$\tilde{u}$.

        \item
         In $T^{\star}$,~$\tilde{u}$ lies on the path between~$w$ and~$\tilde{v}$.
            \end{enumerate}

   We first consider Statement~(i).
   By Observation~\ref{lem-qv-subpath}, $Q(u)\subseteq \vset{P_{\{w, r\}}}$. Since $\tilde{u}\in Q(u)$, we know that~$\tilde{u}$ lies on the path $P_{\{w, r\}}$ between the root~$r$ and the guider~$w$ of~$u$ in~$\sigma$.
   By the definition of~$v$, we have that $w = \guide{v}{\sigma'}$. Recall that all {\free} vertices including~$v$ are leaves of~$\widetilde{T}$.
   Let~$X$ be the set of the first~$k-1$ vertices in~$\sigma$ and~$\sigma'$.
   After visiting the vertices of~$X$,~$\sigma'$ visits~$v$. By Lemma~\ref{lem:undeter-vertex},~$\tilde{v}$ locates on~$P_{\{r,w\}}$  too.
   Since~$\sigma'$ visits~$v$ before~$u$, and~$\tilde{u} \neq \tilde{v}$, by Observation~\ref{obs-xyz}, we have that~$\tilde{v}$ lies on the path between~$w$ and~$\tilde{u}$.

    Now we analyze Statement~(ii).
    We distinguish two cases according to whether~$u$ and~$v$ have the same guider in~$\sigma$.

    Consider first the case where~$\guide{u}{\sigma}= \guide{v}{\sigma}$.
    Since the purification rule has been exhaustively applied, in~$T^{\star}$, every vertex $x \in Q(u) \setminus Q(v)$ is closer to~$w$ than every vertex $y \in Q(v)$, and that every vertex $x \in Q(v) \setminus Q(u)$ is farther to~$w$ than every vertex $y \in Q(u)$. Since~$\widetilde{T}$ is obtained from~$T^{\star}$  by adding the {\free} vertices as leaves, these also hold in~$\widetilde{T}$. Then, if $\tilde{u} \in Q(u) \setminus Q(v)$ or $\tilde{v} \in Q(v) \setminus Q(u)$, the statement in~(ii) holds.
    Otherwise,~$\tilde{u}$ and~$\tilde{v}$ are both from $Q(u) \cap Q(v)$. By the parent assignment phase, we have that $\tilde{u}=\arg \min_{x \in Q(u)} f(x)$ and $\tilde{v}=\arg \min_{x \in Q(v)} f(x)$. It follows that $\tilde{u}=\tilde{v}$, a contradiction, implying that this cannot occur.


    Consider now the second case where $\guide{u}{\sigma}\neq \guide{v}{\sigma}$. By Lemma~\ref{lem-differ-stage}, we have that $| Q(u) \cap Q(v) \le 1 |$ and every vertex $x \in Q(u) \setminus Q(v)$ is closer to $w$ than every vertex of $Q(v)$, and every vertex $x \in Q(u) \cap Q(v)$ is closer to $w$ than every vertex of $Q(v) \setminus Q(u)$.
    The remainder of the proof proceeds analogously to the first case.

Finally, since $\tilde{u} \neq \tilde{v}$ and $T^{\star}$ is a tree, the two statements above contradict each other.
%
 %
\end{proof}

\subsection{Putting All Together}
We have described our polynomial-time algorithm in an expository manner, presenting each step sequentially together with a rigorous proof of its correctness. To provide a clearer overall understanding, this section presents the algorithm in its entirety (Algorithm~\ref{alg-dfs-tree-realize}) and provides a global proof of correctness by modularly combining the individual arguments established for the separate steps.

\begin{algorithm}[ht]
\caption{
}
\label{alg-dfs-tree-realize}
\prealg
{A vertex set $V$ and a set $\Pi$ of orderings of $V$}
{A DFS-tree support of $\Pi$, if one exists; otherwise, return ``\no''}
\begin{algorithmic}[1]
    \State $D\gets \attg{\Pi}$; $G\gets \attug{\Pi}$; \label{alg-label-attachment-graph}
    \If{$G$ is disconnected} \label{alg-label-D-disconnected-condition}
        \State \Return ``\no''; \label{lag-label-D-disconnected}
    \EndIf
    \State let $T^{\star}$ be the subgraph of $G$ induced by the forced vertices of~$D$; \label{alg-label-ini-t-star}
    \State let $V^{\star}\gets \vset{T^{\star}}$;  \label{alg-label-determined} \Comment{set of {\forced} vertices}
    \If{${V^{\star}} = V$} \label{alg-label-t-star-full-a}
        \If{all orderings in $\Pi$ are DFS-orderings of $T^{\star}$}
            \State \Return $T^{\star}$; \label{alg-label-t-star-full-b}
        \Else
            \State \Return ``\no''; \label{alg-label-t-star-full-no}
        \EndIf
    \EndIf

    \If{$\exists \sigma\in \Pi$ whose restriction to ${V^\star}$ is not a DFS-ordering of $T^{\star}$} \label{alg-label-restricted-a}
        \State \Return ``\no''; \label{alg-label-restricted-b}
    \EndIf

    \State $S \gets V \setminus {V^{\star}}$;  \label{alg-label-assign-S}  \Comment{set of {\free} vertices}
    \State compute $Q(v)$ for all $v \in S$ as defined in~\eqref{eq-qv}, and exhaustively apply the Purification Rule; \label{alg-label-purification}
    \If{$Q(v)=\emptyset$ for some $v\in S$} \label{alg-label-Qv-empty-a}
        \State \Return ``\no''; \label{alg-label-Qv-empty-b}
    \EndIf
        \State let $f:V\rightarrow[\abs{V}]$ be an arbitrary bijection;  \label{alg-label-parent-assignment-start}
        \For{each $v \in S$}
            \State let $u \gets \argmin_{x \in Q(v)} f(x)$;
            \State add vertex $v$ and edge $\edge{u}{v}$ to $T^{\star}$; \label{alg-label-parent-assignment-end}
        \EndFor
    \State \Return $T^{\star}$; \label{alg-label-end}
\end{algorithmic}
\end{algorithm}

\begin{theorem}
\label{thm:build-DFS-tree}
Let $\Pi$ be a set of orderings of a vertex set~$V$.
If $\Pi$ admits a DFS-tree support, then Algorithm~\ref{alg-dfs-tree-realize}
 constructs a DFS-tree support of~$\Pi$ in polynomial time;
otherwise, it reports \emph{no}.
\end{theorem}

\begin{proof}
The algorithm first computes the attachment digraph~$D$ and the attachment graph~$G$ of~$\Pi$, using the polynomial-time algorithm of~\cite{PetersYCE22} (Line~\ref{alg-label-attachment-graph}).

If~$\Pi$ admits a DFS-tree support, then it admits a GS-tree support by Lemma~\ref{lem-graph-traversal-relation}. By Lemma~\ref{lem-characterization-gs},~$G$ is connected. Therefore, if~$G$ is disconnected (Line~\ref{alg-label-D-disconnected-condition}), the algorithm correctly returns ``\no'' in Line~\ref{lag-label-D-disconnected}.

In Line~\ref{alg-label-ini-t-star}, the algorithm initializes~$T^{\star}$ as the subgraph of~$G$ induced by all forced vertices in~$D$. Line~\ref{alg-label-determined} defines~$V^{\star}$ as the set of all {\forced} vertices. By Lemmas~\ref{lem-forced-vertex},~$T^{\star}$ is a tree. Furthermore, by Lemmas~\ref{lem-graph-traversal-relation} and~\ref{lem-att-dig-realize}, every DFS-tree support of~$\Pi$ is a spanning tree of~$G$. Therefore, if~$V^{\star}=V$, then~$T^{\star}$ is the unique candidate DFS-tree support of~$\Pi$. Consequently, the algorithm correctly returns ``\yes'' if and only if every ordering in~$\Pi$ is a DFS-ordering of~$T^{\star}$ in Lines~\ref{alg-label-t-star-full-a}--\ref{alg-label-t-star-full-no}.

The algorithm next checks whether $T^{\star}$ is a DFS-tree support of $\Pi_{|{V^{\star}}}$ (Lines~\ref{alg-label-restricted-a} and~\ref{alg-label-restricted-b}). By Observation~\ref{obs-t-star-fail}, if this is not the case, then $\Pi$ admits no DFS-tree support and the algorithm returns ``\no''.

If the algorithm has not terminated so far, then every ordering in~$\Pi$ restricted to~$V^{\star}$ is a DFS-ordering of~$T^{\star}$.
By Lemmas~\ref{lem-dfs-leaves} and~\ref{lem-purify}, if~$\Pi$ admits a DFS-tree support, then  it admits one that contains~$T^{\star}$ as a subtree
and places all {\free} vertices as leaves.

In Line~\ref{alg-label-assign-S}, the algorithm identifies the set~$S$ of all
{\free} vertices. In Line~\ref{alg-label-purification}, it computes
$Q(v)$ for each~$v\in S$ according to~\eqref{eq-qv} and exhaustively applies
the purification rule.

In Lines~\ref{alg-label-Qv-empty-a} and~\ref{alg-label-Qv-empty-b}, the
algorithm returns ``\no'' whenever there exists a vertex~$v\in S$ with
$Q(v)=\emptyset$. By Lemmas~\ref{lem-purify} and~\ref{lem-dfs-leaves}, every
DFS-tree support must assign to each~$v\in S$ a unique neighbor from~$Q(v)$.
Therefore, if $Q(v)=\emptyset$ for some~$v$, $\Pi$ admits no DFS-tree support,
and the algorithm returns ``\no''.
%

If~$Q(v)$ is nonempty for all {\free} vertices~$v\in S$, the
algorithm proceeds to the parent assignment phase
(Lines~\ref{alg-label-parent-assignment-start}--\ref{alg-label-parent-assignment-end}),
whose correctness is established in Lemma~\ref{lem-correctness-assignment},
and outputs~$T^{\star}$ in Line~\ref{alg-label-end}.

Since both the exhaustive application of the purification rule and the parent
assignment phase can be performed in polynomial time, the entire algorithm runs
in polynomial time.
\end{proof}

{
\section{Concluding Remarks}
In this paper, we introduced a unified framework for studying structured preferences through graph search paradigms. We investigated whether a given preference profile admits a graph support whose traversals generate the observed preference orderings under six fundamental graph search paradigms: BFS, DFS, LexBFS, LexDFS, MCS, and MNS. We established NP-hardness results for both \textsc{Edge-Bounded-S-SP} and \textsc{Deg-Bounded-S-SP} for all considered paradigms, thereby extending the previously known hardness results for generic search. For the tree-support variant, we developed a polynomial-time algorithm for recognizing DFS-tree supports, resolving one of the two previously open cases and establishing several structural properties of DFS on trees that may be of independent interest. Together with existing results, our results yield a nearly complete complexity classification of graph-search-based preference domains, with the complexity of recognizing BFS-tree supports remaining the only open case.

\begin{openquestion}
    What is the complexity of {\prob{SP-$\mathcal{S}$-Tree}}.
\end{openquestion}

Our results also motivate a closer examination of graph-search-based preference domains on more restricted graph classes.
On paths and cycles, GS-, MCS-, and MNS-orderings coincide. Consequently, MCS- and MNS-compatible preferences on paths (respectively, cycles) coincide exactly with single-peaked preferences (respectively, single-peaked preferences on a circle), which can be determined in polynomial time~\cite{PetersYCE22}.

The remaining cases admit the following simple characterization.

\begin{theorem}
For every~$\mathcal{S}\in\{\mathrm{DFS},\mathrm{LexDFS},\mathrm{BFS},\mathrm{LexBFS}\}$, the following recognition problems are polynomial-time solvable:
\begin{itemize}
    \item Determining whether a set~$\Pi$ of orderings admits an~$\mathcal{S}$-support that is a path.
    \item Determining whether a set~$\Pi$ of orderings admits an~$\mathcal{S}$-support that is a cycle.
\end{itemize}
\end{theorem}

\begin{proof}
Consider the first problem for DFS. If there are more than two vertices appearing last in the given orderings, output {\no}. If exactly two such vertices exist, place them at the two ends of the path and determine the remaining vertices greedily. If there is only one such vertex, place it at one end of the path and again determine the remaining vertices greedily.

For the first problem for BFS, observe that for any single ordering, there is a unique path that admits it as a BFS-ordering, which can be identified in polynomial time. Therefore, to solve the problem, we first arbitrarily select an ordering in $\Pi$, and then verify whether every other ordering in $\Pi$  is a BFS-ordering of the path determined by the selected ordering.

The algorithms for the second problem are analogous to those for the first. Since LexDFS and DFS (respectively, LexBFS and BFS) coincide on paths and cycles, the corresponding problems for LexDFS and LexBFS are also polynomial-time solvable.
\end{proof}
}

\section*{Acknowledgements}

This work was supported by the National Natural Science Foundation of China (Grant No.~62302060).


\end{document}